\documentclass[11pt,english]{article}
\usepackage{ae,aecompl}
\usepackage[T1]{fontenc}
\usepackage[utf8]{inputenc}
\usepackage{color}
\usepackage{array}
\usepackage{float}
\usepackage{booktabs}
\usepackage{amsmath}
\usepackage{amsthm}
\usepackage{amssymb}
\usepackage{graphicx}
\usepackage{geometry}
\usepackage{tablefootnote}
\usepackage{setspace}
\usepackage[authoryear]{natbib}
\PassOptionsToPackage{normalem}{ulem}
\usepackage{ulem}
\makeatletter

\providecommand{\tabularnewline}{\\}
\floatstyle{ruled}
\newfloat{algorithm}{tbp}{loa}
\providecommand{\algorithmname}{Algorithm}
\floatname{algorithm}{\protect\algorithmname}

\theoremstyle{plain}
\newtheorem{prop}{\protect\propositionname}
\theoremstyle{plain}
\newtheorem{lem}{\protect\lemmaname}
\ifx\proof\undefined\
  \newenvironment{proof}[1][\proofname]{\par
    \normalfont\topsep6\p@\@plus6\p@\relax
    \trivlist
    \itemindent\parindent
    \item[\hskip\labelsep
          \scshape
      #1]\ignorespaces
  }{%
    \endtrivlist\@endpefalse
  }
  \providecommand{\proofname}{Proof}
\fi

\usepackage[bottom]{footmisc}
\usepackage{setspace}
\usepackage{caption}

\DeclareSymbolFont{extrasymbols}{OMS}{cmsy}{m}{n}
\DeclareMathDelimiter{\lVert}
  {\mathopen}{extrasymbols}{"6B}{largesymbols}{"0D}
\DeclareMathDelimiter{\rVert}
  {\mathclose}{extrasymbols}{"6B}{largesymbols}{"0D}

\makeatother

\usepackage{babel}
\providecommand{\lemmaname}{Lemma}
\providecommand{\propositionname}{Proposition}

\begin{document}
\title{{\Large Estimating Parameters of Structural Models Using Neural Networks\vspace{1.5em}}}
\author{Yanhao 'Max' Wei and Zhenling Jiang\thanks{Marshall School of Business, University of Southern California and
Wharton School, University of Pennsylvania. Email: yanhaowe@usc.edu;
zhenling@wharton.upenn.edu. We thank the constructive comments from
Eric Bradlow, Tat Chan, Vineet Kumar, Greg Lewis, Laura Liu, Nitin
Mehta, Aviv Nevo, Holger Sieg, Andrey Simonov, K. Sudhir, Raphael
Thomadsen, Kosuke Uetake, Shunyuan Zhang, and the participants at
the Northwestern marketing seminar, UCLA marketing seminar, Columbia
marketing camp, Junior faculty forum at WashU, Yale marketing seminar,
Rochester marketing seminar, Tech Adoption and Human-Algorithm Interaction
Workshop at Harvard, 42nd Marketing Science Conference, QME 2020 Conference,
and SICS 2021. Max thanks Sam Boysel for research assistance.}}

\date{2024-08-09}
\maketitle

\begin{abstract}
\begin{singlespace}We study an alternative use of machine learning.
We train neural nets to provide the parameter estimate of a given
(structural) econometric model, e.g., discrete choice or consumer
search. Training examples consist of datasets generated by the econometric
model under a range of parameter values. The neural net takes the
moments of a dataset as input and tries to recognize the parameter
value underlying that dataset. Besides the point estimate, the neural
net can also output statistical accuracy. This neural net estimator
(NNE) tends to limited-information Bayesian posterior as the number
of training datasets increases. We apply NNE to a consumer search
model. It gives more accurate estimates at lighter computational costs
than the prevailing approach. NNE is also robust to redundant moment
inputs. In general, NNE offers the most benefits in applications where
other estimation approaches  require very heavy simulation costs.
We provide code at: \texttt{https://nnehome.github.io}. 

\vspace{1em}


\end{singlespace}
\end{abstract}
\pagebreak

\section{Introduction}

Machine learning algorithms are increasingly applied in empirical
research of marketing and economics. There are two streams of applications:
(i) processing traditionally intractable data, such as texts and images,
and (ii) building more flexible economic models, e.g., predicting
user choices based on user attributes and histories.\footnote{See, e.g., \citet{timoshenko2019identifying}, \citet{liu2019large},
\citet{chiong2019random}, \citet{zhu2020dynamic}, \citet{hema2020personalization},
\citet{zhang2022can}, \citet{yoganarasimhan2022design}. Also see
\citet{athey2018impact} for a review.} At the core of these applications is the exceptional capability of
machine learning algorithms to learn functions. An example is object
recognition from images, which uses neural nets to learn a mapping
from image pixels to the object in the image (e.g., aircraft, boat,
cat).

This paper studies an alternative use of machine learning algorithms
that capitalizes on their ability to learn functions. We train a machine
learning model to ``recognize'' the parameter value of a given econometric
model (e.g., discrete choice, consumer search, games). In this paper,
we focus on shallow neural nets as the machine learning model.\footnote{In principle, one may choose other machine learning models. We find
neural nets to perform well in applications.} We construct training examples by using the econometric model to
simulate datasets under a range of parameter values. The neural net
takes a dataset as input and tries to recognize the parameter value
underlying that dataset. We refer to this approach as the neural net
estimator (NNE).

As such, NNE is a tool that exploits existing machine learning techniques
to estimate current econometric models. It is particularly helpful
in structural estimation. Structural models are useful for empirical
research (e.g., enabling counterfactuals) but can be very difficult
to estimate. Particularly, as larger and richer data allow for more
complex structural models, many of them require increasingly heavy
simulations to evaluate the objective function for estimation. The
objective function can also be difficult to optimize (e.g., non-smooth)
or code. As we will show, in such problems NNE can provide large computational
savings and more accurate estimates, which makes it a useful addition
to the toolbox for empirical research.

Section \ref{sec:front} formulates NNE in detail. We start with a
fairly general setting of structural estimation. Let $i\in1,...,n$
index the individuals in data (which may or may not be i.i.d.) Let
$\boldsymbol{x}_{i}$ denote the observed attributes of $i$ and $\boldsymbol{y}_{i}$
denote the outcome of $i$. Let $\boldsymbol{x}\equiv\{\boldsymbol{x}_{i}\}_{i=1}^{n}$
and $\boldsymbol{y}\equiv\{\boldsymbol{y}_{i}\}_{i=1}^{n}$. So $\{\boldsymbol{y},\boldsymbol{x}\}$
denotes a dataset. An econometric model specifies $\boldsymbol{y}$
as a function of $\boldsymbol{x}$, a set of unobservables $\boldsymbol{\varepsilon}$,
and parameter $\boldsymbol{\theta}$. We write $\boldsymbol{y}=\boldsymbol{q}(\boldsymbol{x},\boldsymbol{\varepsilon};\boldsymbol{\theta})$.
Examples of $\boldsymbol{q}$ are discrete choices, games, diffusion
on networks, etc. To estimate $\boldsymbol{\theta}$, common approaches
include MLE and GMM. In many applications, integrals over the unobservable
$\boldsymbol{\varepsilon}$ are evaluated via simulations, which gives
rise to simulated maximum likelihood (SMLE) and simulated method of
moments (SMM).

NNE offers an alternative approach to estimate $\boldsymbol{\theta}$.
We train a neural net to recognize the value of $\boldsymbol{\theta}$
when it is given a dataset generated by the econometric model under
that value of $\boldsymbol{\theta}$. To construct training examples,
NNE only requires that it is possible to simulate the outcome $\boldsymbol{y}$
using the econometric model. Specifically, let $\ell=1,...,L$ index
training examples. For each $\ell$, we draw $\boldsymbol{\theta}^{(\ell)}$
from a parameter space. Then, we draw $\boldsymbol{\varepsilon}^{(\ell)}$
and compute $\boldsymbol{y}^{(\ell)}$ by the econometric model, i.e.,
$\boldsymbol{y}^{(\ell)}=\boldsymbol{q}(\boldsymbol{x},\boldsymbol{\varepsilon}^{(\ell)};\boldsymbol{\theta}^{(\ell)})$.
The collection $\{\boldsymbol{\theta}^{(\ell)},\{\boldsymbol{y}^{(\ell)},\boldsymbol{x}\}\}$
constitutes a training example. An illustration is given in Figure
\ref{fig:illustrate_train_set}. Additional examples can be generated
for validation purposes. Finally, we apply the trained neural net
to the real data to get an estimate for $\boldsymbol{\theta}$.

\begin{figure}
\begin{centering}
\includegraphics[scale=0.25]{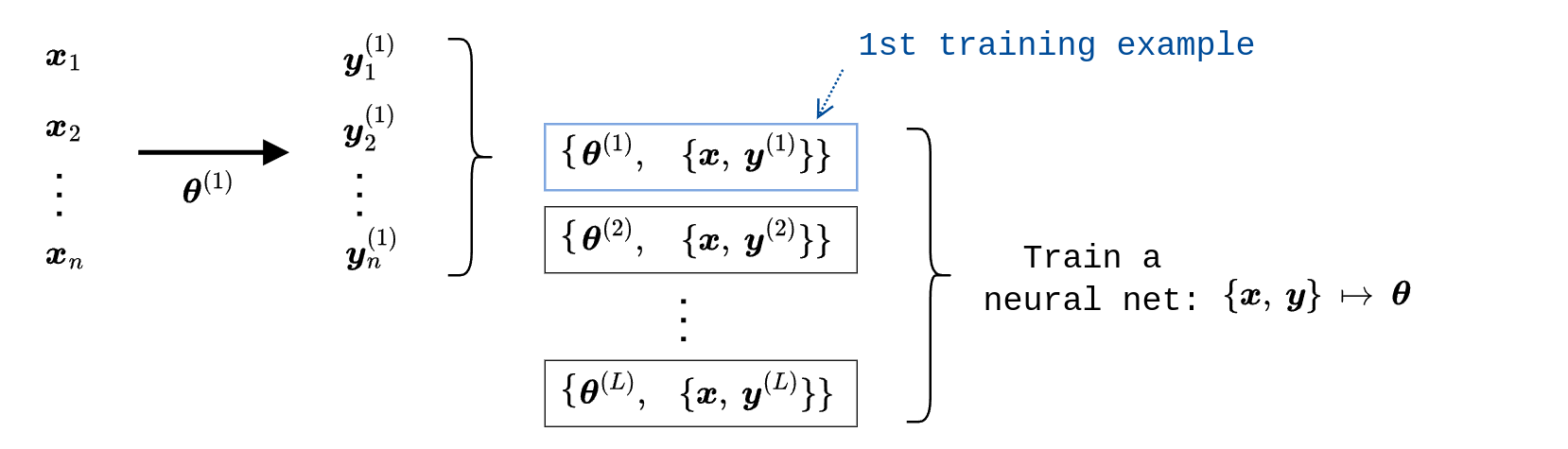}
\par\end{centering}
\caption{Training Examples in NNE\label{fig:illustrate_train_set}}
\end{figure}

We focus on two variations of the above approach. First, we allow
using data moments (instead of the entire dataset) as the input to
the neural net. Examples of moments include the mean of $\boldsymbol{y}_{i}$,
cross-covariance between $\boldsymbol{y}_{i}$ and $\boldsymbol{x}_{i}$,
etc. Economic intuition often provides guidance for choosing moments.
The reduction of a dataset into moments allows us to use simpler neural
nets that are computationally easy to train. Second, we configure
the neural net so that it outputs not only point estimate of $\boldsymbol{\theta}$
but also a measure of statistical accuracy for the point estimate.

In Section \ref{subsec:front-convergence}, we borrow from the econometric
literature on properties of neural nets to provide several theoretical
properties for NNE. Let $\boldsymbol{m}$ collect the data moments.
As the size of training set $L\rightarrow\infty$, the point estimate
by NNE converges to $\mathbb{E}(\boldsymbol{\theta}|\boldsymbol{m})$.
In Bayesian language, $\mathbb{E}(\boldsymbol{\theta}|\boldsymbol{m})$
is known as the limited-information posterior mean. Here, ``limited-information''
signifies that the posterior is conditioned on the moments instead
of the entire data. Further, the measure of statistical accuracy given
by NNE converges to $\mathbb{V}\text{ar}(\boldsymbol{\theta}|\boldsymbol{m})$
or $\mathbb{C}\text{ov}(\boldsymbol{\theta}|\boldsymbol{m})$.

One implication of these results is robustness to redundant moments.
To see it, note that in a given application, the accuracy of $\mathbb{E}(\boldsymbol{\theta}|\boldsymbol{m})$
in estimating $\boldsymbol{\theta}$ will not decrease as $\boldsymbol{m}$
includes more moments. Thus, with a sufficiently large $L$, NNE will
not become less accurate as $\boldsymbol{m}$ includes more moments.
This is different from SMM, where even with sufficiently many simulations,
redundant moments are known to increase finite-sample biases.\footnote{See \citet{Xu2015selectmoments} and references there within. Also
see \citet{newey2007gmmnotes}.} Intuitively, this difference is because NNE can learn from training
examples whether a moment contributes to recognizing $\boldsymbol{\theta}$.
SMM and GMM lack such a mechanism (optimal weighting matrix accounts
for only the variance-covariance \textit{among} moments). As such,
NNE gives researchers flexibility to include more moments for estimation.

We study NNE in two applications. Section \ref{sec:AR1} uses a simple
AR(1) model to visualize the working of NNE and its differences from
SMM. The simplicity of the AR(1) model also allows us to see whether
a moment is redundant. As we include redundant or almost redundant
moments, the root mean square error (RMSE) of GMM increases substantially.
In contrast, the increases in the RMSE of NNE are much smaller.

Section \ref{sec:search} studies an application to sequential search
model. The search model has gained popularity as data on consumer
online search journeys become increasingly available. The prevailing
method to estimate search model is SMLE. However, the estimation is
challenging because the number of possible search and purchase combinations
is very large. Evaluating the probability of any particular combination
requires an infeasibly large number of model simulations. To make
SMLE feasible, studies in literature have applied smoothing to the
likelihood function.

In Monte Carlo studies, we find that NNE recovers the search model
parameter well. Over a wide range of computational costs, NNE has
a smaller RMSE than SMLE. The advantage of NNE is especially large
at low computational costs. In addition, SMLE's estimates are sensitive
to the degree of smoothing. The sensitivity makes it particularly
important to choose the appropriate smoothing factor in SMLE. However,
the choice of smoothing factor lacks established guidance and has
been known to be very difficult for consumer search models in general.
When we apply both NNE and SMLE to a real dataset of consumer search,
the search model estimated by NNE produces a better fit to key data
statistics. The estimates and model fit by SMLE are again sensitive
to smoothing.

We conduct several additional analyses of NNE in the context of the
search model. First, we find that the estimation accuracy of NNE does
not deteriorate as $\boldsymbol{m}$ includes more moments, which
again suggests robustness to redundant moments. Second, the accuracy
of NNE improves with $n$ faster than smoothed SMLE. Third, the estimate
by NNE is not sensitive to neural net configuration, and the optimal
configuration can be easily selected using validation examples. Fourth,
NNE remains reasonably well-behaved even when $\Theta$ fails to include
the true value of $\boldsymbol{\theta}$; the estimates mostly fall
between $\Theta$ and the true value.

Overall, our results suggest that NNE is mostly suitable for applications
where otherwise heavy simulations are needed to integrate out the
unobservables for estimation. In these cases, NNE provides a computationally
light alternative to obtain accurate estimates. For applications where
the main estimation burden is not simulations, NNE may still be applied
but is unlikely to show clear gains. Examples include applications
where closed-form likelihood is available and dynamic choices or games
where solving the economic model constitutes the main burden in estimation. 

We summarize the takeaways of NNE in Table \ref{tab:pros-and-cons}.
To facilitate applications, we provide our code at: \texttt{https://nnehome.github.io}.
Finally, we discuss several extensions for future research in Section
\ref{sec:conclusion}. In particular, further development of NNE may
make it practical to rely on aggregate moments to estimate certain
structural models that normally require individual-level data. The
extension will allow for privacy-preserving structural estimation.

\subsection{Literature}

This paper adds to a fast-growing literature at the intersection of
machine learning and marketing/economics. Machine learning algorithms
have been applied to various marketing and economic questions. A particular
stream of research focuses on leveraging machine learning to tackle
estimation problems. \citet{lewis2018adversarial} use the adversarial
approach to select from a continuum of moment conditions. \citet{wager2018estimation}
apply random forests to estimate heterogeneous treatment effects.
\citet{Chernozhukov2018double} study using machine learning models
to capture high-dimensional nuisance parameters. \citet{farrell2020deep}
use neural nets to model individual heterogeneity. The literature
so far has focused on applying machine learning to build more flexible
models between economic variables. Our paper explores a different
direction, which uses machine learning to estimate the parameter of
an econometric model. A recent work sharing this direction is \citet{kaji2020adversarial}
that apply the adversarial approach. They propose training a discriminator
to construct the objective function of extremum estimator. Their approach
is very different from ours (e.g., NNE is not an extremum estimator).

\begin{table}
\caption{Summary of NNE\label{tab:pros-and-cons}}

\begin{centering}
\begin{tabular}{>{\raggedleft}p{0.02\textwidth}>{\raggedright}p{0.92\textwidth}}
\toprule 
 & {\small\hspace{18em}}{\small\textbf{Main properties}}\tabularnewline
\midrule
{\small 1.} & {\small It does not require computing integrals over unobservables
in the structural econometric model; it only requires being able to
simulate data using the econometric model.}\tabularnewline\addlinespace[1mm]
{\small 2.} & {\small It does not require optimizing over an (potentially non-smooth)
objective function as in extremum estimators (e.g., SMLE, SMM, indirect
inference).}\tabularnewline\addlinespace[1mm]
{\small 3.} & {\small It is more robust to redundant moments when compared to SMM/GMM.}\tabularnewline\addlinespace[1mm]
{\small 4.} & {\small It computes a measure of statistical accuracy} {\small as a
byproduct.}\tabularnewline\addlinespace[1mm]
\bottomrule
\end{tabular}
\par\end{centering}
\medskip{}

\centering{}%
\begin{tabular}{>{\raggedright}p{0.47\textwidth}>{\raggedright}p{0.47\textwidth}}
\toprule 
{\small\hspace{5em}}{\small\textbf{Suitable applications}} & {\small\hspace{5em}}{\small\textbf{Less suitable applications}}\tabularnewline
\midrule 
{\small A large number of simulations are needed to evaluate likelihood/moments.
The SMLE/SMM objective is difficult to optimize. There lacks clear
guidance on moment choice. Formulas for standard errors are not yet
established.}{\small\tablefootnote{An example is data with network dependence, where asymptotic theory
is an active research area. See \citet{lee2021networkinference} and
\citet{graham2020networkdata} for surveys. In this case, NNE provides
an alternative for inferences.}} & {\small Closed-form expressions are available for likelihood/moments.
The main estimation burden comes from sources other than the simulations
to evaluate likelihood/moments.}\tabularnewline
{\small\uline{Examples}}{\small : sequential search, discrete choices
with rich unobserved heterogeneity, matching games, choices on networks.} & {\small\uline{Examples}}{\small : dynamic choices or games where the
main burden is solving policy functions.}\tabularnewline
\bottomrule
\end{tabular}
\end{table}

We borrow from the econometric literature on neural nets to derive
theoretical properties of NNE. This literature started with \citet{white1989learning}
and \citet{white1990connectionist}, which establish that single-layer
neural nets can learn arbitrary conditional expectations. Later studies
mostly examine neural nets under the class of sieve estimators (see
\citealt{chen2007sieve}). More recently, \citet{farrel2019convergence}
study multi-layer neural nets. This literature assumes the setting
where the goal is to learn relations between data variables. Our paper
assumes a different setting where the goal is to learn the relation
between data and an econometric model's parameter. We adapt the results
in the literature and derive conditions under which they can continue
to hold in our setting.

This paper joins a literature that develops computationally light
estimators (e.g., \citealt{bajari2007estimating}, \citealt{pakes2007simple},
\citealt{su2012constrained}). A particular estimator which NNE shares
a similarity with is indirect inference (\citealt{gourieroux1993indirect},
\citealt{collard-wexler2013indirect}, \citealt{bao2017could}). Both
approaches avoid the integrals over unobservables in the econometric
model. However, the two approaches are conceptually different. As
discussed in \citet{gourieroux1993indirect} and \citet{smith2008indirect},
indirect inference is an extension of SMM, with the usual moments
replaced by the parameter of an auxiliary model. As a result, it falls
under extremum estimators and inherits issues from SMM. For example,
indirect inference has a non-smooth objective function in discrete
choice models, which creates difficulty for optimization (\citealt{bruins2018generalized}).
NNE is not an extremum estimator and faces no such difficulty. In
addition, similar to how SMM can be sensitive to moment choice, indirect
inference can be sensitive to the choice of auxiliary model. NNE puts
less burden on researchers to select moments. We discuss indirect
inference further in Appendix \ref{subsec:indirect-infer}.

\vspace{1.2em}

\section{Learning parameter value from data\label{sec:front}}

In this section, we start by describing a framework of structural
estimation. Under this framework, we describe how NNE computes the
parameter estimate of the structural econometric model and the statistical
accuracy of the estimate. We also provide the limit of NNE as $L$
increases.

\subsection{Setup}

In most applications, we are interested in modeling an outcome $\boldsymbol{y}_{i}$,
with $i=1,...,n$ indexing the individuals in data. For example, in
a consumer search model, $\boldsymbol{y}_{i}$ includes both the search
and purchase choices of consumer $i$. We also observe some attributes
$\boldsymbol{x}_{i}$. In the consumer search model, $\boldsymbol{x}_{i}$
collects the attributes of the products available to consumer $i$.
Let $\boldsymbol{y}\equiv\{\boldsymbol{y}_{i}\}_{i=1}^{n}$ and $\boldsymbol{x}\equiv\{\boldsymbol{x}_{i}\}_{i=1}^{n}$
collect the data across all individuals. For the purposes of NNE,
we consider econometric models that are parametric and allow researchers
to simulate the outcome $\boldsymbol{y}$. Formally, denote the econometric
model as $\boldsymbol{q}$ and 
\begin{equation}
\boldsymbol{y}=\boldsymbol{q}(\boldsymbol{x},\boldsymbol{\varepsilon};\boldsymbol{\theta}),\label{eq:econ-model}
\end{equation}
where $\boldsymbol{\varepsilon}$ collects the unobserved attributes
(sometimes interpreted simply as error terms), and $\boldsymbol{\theta}$
denotes the parameter vector (finite-dimensional in this paper). Throughout
this paper, the word ``model'' when used alone always refers to
an econometric model. When referring to a machine learning ``model,''
we will specifically point it out.

This formulation includes cases with i.i.d. and non-i.i.d. observations.
When individuals are independent observations, $\boldsymbol{y}_{i}$
does not depend on $(\boldsymbol{x}_{j},\boldsymbol{\varepsilon}_{j})$
for any $j\neq i$. However, for example, in a model of peer influence
over social network, choices are inter-dependent. So $\boldsymbol{y}_{i}$
may depend on $(\boldsymbol{x}_{j},\boldsymbol{\varepsilon}_{j})$
for $j\neq i$ and $\boldsymbol{y}_{i}$ is generally a function of
the entire $(\boldsymbol{x},\boldsymbol{\varepsilon})$.\footnote{For example, let $\boldsymbol{A}$ be the adjacency matrix and $y_{i}=\lambda\sum_{j=1}^{n}A_{ij}y_{j}+\boldsymbol{\beta}'\boldsymbol{x}_{i}+\varepsilon_{i}$,
then $\boldsymbol{y}=(\boldsymbol{I}-\lambda\boldsymbol{A})^{-1}(\boldsymbol{x}\boldsymbol{\beta}+\boldsymbol{\varepsilon})$.}

The maximum likelihood estimator (MLE) estimates $\boldsymbol{\theta}$
by maximizing $\log P(\boldsymbol{y}|\boldsymbol{x};\boldsymbol{\theta})$.
In the i.i.d. case, we have $\log P(\boldsymbol{y}|\boldsymbol{x};\boldsymbol{\theta})=\frac{1}{n}\sum_{i=1}^{n}\log P(\boldsymbol{y}_{i}|\boldsymbol{x}_{i};\boldsymbol{\theta})$.
For many econometric models, there is no explicit expression for $P(\boldsymbol{y}_{i}|\boldsymbol{x}_{i};\boldsymbol{\theta})$,
and one has to use Monte Carlo simulations to integrate out $\boldsymbol{\varepsilon}_{i}$.
This is known as the simulated maximum likelihood (SMLE). For subsequent
discussions, it is useful to define simulation burden. Fix a data
size $n$. Suppose the researcher uses $R$ draws of $\boldsymbol{\varepsilon}_{i}$
to evaluate $P(\boldsymbol{y}_{i}|\boldsymbol{x}_{i};\boldsymbol{\theta})$.
An optimization routine needs to evaluate the likelihood $P(\boldsymbol{y}|\boldsymbol{x};\boldsymbol{\theta})$
many times at different values of $\boldsymbol{\theta}$. We define
the simulation burden as $R$ times the number of likelihood evaluations
by the optimization routine.

An alternative method to estimate $\boldsymbol{\theta}$ is the generalized
method of moments (GMM). As with MLE, for many econometric models
the evaluation of the moment function requires simulations to integrate
out $\boldsymbol{\varepsilon}_{i}$. This is known as the simulated
method of moments (SMM). Suppose the researcher uses $R$ draws of
$\boldsymbol{\varepsilon}_{i}$ to evaluate the moments. Similarly,
the simulation burden is $R$ times the number of moment function
evaluations by the optimization routine.

\subsection{NNE point estimates \label{subsec:front-point}}

Fix an econometric model $\boldsymbol{q}$ as in equation (\ref{eq:econ-model}).
Recall that $\boldsymbol{y}\equiv\{\boldsymbol{y}_{i}\}_{i=1}^{n}$
and $\boldsymbol{x}\equiv\{\boldsymbol{x}_{i}\}_{i=1}^{n}$, thus
$\{\boldsymbol{y},\boldsymbol{x}\}$ denotes a dataset. In what follows,
we describe NNE as using a machine learning algorithm to learn a mapping
from datasets to the parameter values of the econometric model $\boldsymbol{q}$:
\begin{equation}
\{\boldsymbol{y},\boldsymbol{x}\}\mapsto\boldsymbol{\theta}.\label{eq:front-mapping}
\end{equation}
That is, given any dataset, NNE tries to ``recognize'' the value
of $\boldsymbol{\theta}$ for the econometric model. Importantly,
the purpose of the machine learning algorithm is to estimate $\boldsymbol{\theta}$,
not to model how $\boldsymbol{x}$ affects $\boldsymbol{y}$. The
relation between $\boldsymbol{y}$ and $\boldsymbol{x}$ is entirely
specified by the econometric model $\boldsymbol{q}$.

Building the machine learning model requires a training set composed
of training examples. Many machine recognition tasks (e.g., image
classification) require laborious human labeling to create training
sets. This is not required in NNE. Instead, we use the econometric
model to generate the training set.

Specifically, let $\ell\in\{1,2,...,L\}$ index the training examples.
For each $\ell$, draw $\boldsymbol{\theta}^{(\ell)}$ uniformly from
a parameter space $\Theta$.\footnote{We give more discussion on how to choose $\Theta$ in practice in
Section \ref{sec:search} and Appendix \ref{subsec:range}.} One may use a non-uniform distribution if a specific prior is available.
Let $\boldsymbol{\varepsilon}^{(\ell)}$ be a random draw of the unobservables
in the econometric model and $\boldsymbol{y}^{(\ell)}=\boldsymbol{q}(\boldsymbol{x},\boldsymbol{\varepsilon}^{(\ell)};\boldsymbol{\theta}^{(\ell)})$.
That is, $\{\boldsymbol{y}^{(\ell)},\boldsymbol{x}\}$ is a dataset
generated by the econometric model given $\boldsymbol{\theta}^{(\ell)}$
and the observed $\boldsymbol{x}$. We call it a training dataset.
Note that all training datasets are generated conditional on the same
observed $\boldsymbol{x}$. Our training set is
\begin{equation}
\{\boldsymbol{\theta}^{(\ell)},\{\boldsymbol{y}^{(\ell)},\boldsymbol{x}\}\}_{\ell=1,2...,L}.\label{eq:front-train-set-D}
\end{equation}
We train a machine learning model to recognize $\boldsymbol{\theta}^{(\ell)}$
from $\{\boldsymbol{y}^{(\ell)},\boldsymbol{x}\}$. After training,
we plug the real data into the machine learning model to obtain the
estimate for $\boldsymbol{\theta}$.

The machine learning model above uses an entire dataset as the input,
so the input dimension is in the order of $n$. One possible variation
is to use data moments as the input instead. For many structural econometric
models, data moments are sufficient for pinning down parameters, and
economic theory and intuition often suggest what these moments might
be. Let $\boldsymbol{m}$ denote a set of moments that summarizes
$\{\boldsymbol{y},\boldsymbol{x}\}$. For example, $\boldsymbol{m}$
may collect the mean of $\boldsymbol{y}_{i}$: $\bar{\boldsymbol{y}}\equiv n^{-1}\sum_{i=1}^{n}\boldsymbol{y}_{i}$,
the covariance matrix of $\boldsymbol{y}_{i}$: $n^{-1}\sum_{i=1}^{n}(\boldsymbol{y}_{i}-\bar{\boldsymbol{y}})(\boldsymbol{y}_{i}-\bar{\boldsymbol{y}})'$,
the cross-covariance matrix between $\boldsymbol{y}_{i}$ and $\boldsymbol{x}_{i}$:
$n^{-1}\sum_{i=1}^{n}(\boldsymbol{y}_{i}-\bar{\boldsymbol{y}})(\boldsymbol{x}_{i}-\bar{\boldsymbol{x}})'$,
etc.\footnote{Here, it is useful to note that we need not include moments that involve
$\boldsymbol{x}$ only (e.g., covariance matrix of $\boldsymbol{x}_{i}$),
because such moments would be constant across the training examples.} Similarly, let $\boldsymbol{m}^{(\ell)}$ denote the moments that
summarize the simulated dataset $\{\boldsymbol{y}^{(\ell)},\boldsymbol{x}\}$.
Our training set now becomes
\begin{equation}
\{\boldsymbol{\theta}^{(\ell)},\boldsymbol{m}^{(\ell)}\}_{\ell=1,2...,L}.\label{eq:front-train-set-M}
\end{equation}
We train a machine learning model to recognize $\boldsymbol{\theta}^{(\ell)}$
from $\boldsymbol{m}^{(\ell)}$. After training, we plug in the real
data moments to obtain the estimate for $\boldsymbol{\theta}$.

The reduction of data to moments has pros and cons. The dimension
of $\boldsymbol{m}$ is much smaller than that of $\{\boldsymbol{y},\boldsymbol{x}\}$.
In addition, the mapping from moments to $\boldsymbol{\theta}$ is
usually less complex than the mapping from data to $\boldsymbol{\theta}$.
Together, these factors reduce the required size and complexity of
our machine learning model, which lowers training cost. While large
and complex machine learning models (e.g., deep learning models) are
common nowadays, they tend to be costly to train.

Using data moments can result in some loss of information in the data.
At minimum, we would require that the parameter $\boldsymbol{\theta}$
is identified by the moments selected into $\boldsymbol{m}$ (in this
paper we do not specifically address cases where identification is
lacking). This requirement is the same as in GMM, where the moment
conditions must identify parameter. To mitigate the potential information
loss, one would like to include as many potentially relevant moments
as possible. However, it is known that in GMM redundant moments can
substantially worsen finite-sample estimation biases (e.g., \citealt{Xu2015selectmoments},
\citealt{newey2007gmmnotes}). This problem has been studied in settings
with lagged moments (\citealt{andersen1996manymoments}, \citealt{altonji1996small}),
higher-order moments (\citealt{breusch1999momentredundancy}), and
interaction moments (\citealt{donald2021numberofmoments}). In such
settings, researchers have to carefully balance the information gain
against the adverse effect on bias when deciding whether to include
additional moments for estimation.

Fortunately, we find NNE to be more robust to redundant moments. We
first reveal this property from a theoretical perspective in Section
\ref{subsec:front-convergence}, and then illustrate it with applications
in Section \ref{sec:AR1} and \ref{sec:search}. Intuitively, the
robustness comes from the fact that the machine learning algorithm
can learn, from the training set, which moments contribute to predicting
the parameter and which moments do not. GMM lacks such a mechanism
(optimal weighting matrix accounts for the variance-covariance among
moments, not how much each moment contributes to parameter estimation).

We now turn to more details on the machine learning model. Use $\boldsymbol{f}:\boldsymbol{m}\mapsto\boldsymbol{\theta}$
to denote a generic mapping from the space of moments to the space
of parameter values. We seek to find a mapping $\widehat{\boldsymbol{f}}$
such that
\[
\boldsymbol{\theta}^{(\ell)}\simeq\widehat{\boldsymbol{f}}(\boldsymbol{m}^{(\ell)}),\;\forall\ell.
\]

One may use any machine learning algorithm here. In this paper, we
construct $\widehat{\boldsymbol{f}}$ using shallow neural nets. The
main reasons are: (i) in our context shallow neural nets provide a
good balance between learning capacity and ease of training, and (ii)
there are relatively well developed statistical properties of shallow
neural nets for us to establish convergence results (later in Section
\ref{subsec:front-convergence}).

Algorithms for training neural nets are well established and we shall
not expand on their details. For the discussion here, it is sufficient
to view neural nets as a flexible functional form for $\boldsymbol{f}$
parameterized by a set of ``weights.'' The training chooses these
weights by minimizing a loss function. Let $k$ index the dimensions
of $\boldsymbol{\theta}$. We first consider the mean squared error
(MSE) loss function:
\begin{equation}
C_{1}(\boldsymbol{f})=L^{-1}\sum_{\ell=1}^{L}\sum_{k}\left[f_{k}(\boldsymbol{m}^{(\ell)})-\theta_{k}^{(\ell)}\right]^{2}.\label{eq:front-loss-mse}
\end{equation}
We choose $\widehat{\boldsymbol{f}}$ as the neural net $\boldsymbol{f}$
that minimizes $C_{1}$.\footnote{We abstract away from optimization issues in neural net training and
assume that the minimum can be attained. Neural network optimization
is an active research area (see e.g., \citealt{dauphin2014saddlepoint},
\citealt{li2018landscape}, \citealt{Du2019gradientdescent}).} We will consider alternative loss functions later.\textcolor{black}{{}
Readers familiar with machine learning may note that loss function
(\ref{eq:front-loss-mse}) does not have the usual penalty term for
regularization. Instead, we regularize the neural net by choosing
the number of hidden nodes. We do so for the convenience to establish
convergence results (}Section \ref{subsec:front-convergence}), as
the econometric literature on neural nets has mostly focused on regularization
via the number of hidden nodes (instead of a penalty term). In applications,
however, our experience is that NNE works well with either regularization
method.

The final element to be discussed is validation. Specifically, we
simulate additional datasets $\ell=L+1,...,L^{*}$ to form a validation
set $\{\boldsymbol{\theta}^{(\ell)},\boldsymbol{m}^{(\ell)}\}_{\ell=L+1}^{L^{*}}$.
Note the index here starts with $L+1$. The validation set is not
used in the loss function (\ref{eq:front-loss-mse}). Instead, the
validation set can be used to check the accuracy of $\widehat{\boldsymbol{f}}$
as well as choose the number of hidden nodes.

\subsection{NNE statistical accuracy\label{subsec:front-accuracy}}

The development so far has focused on training the neural net to give
a point estimate of $\boldsymbol{\theta}$. Next, we train the neural
net to provide the statistical accuracy of point estimates. Importantly,
we will use the same training set as in expression (\ref{eq:front-train-set-M}).
At first glance, this task may appear implausible because in each
training example, there is only one parameter value $\boldsymbol{\theta}^{(\ell)}$
associated with one moment value $\boldsymbol{m}^{(\ell)}$. However,
across the training examples, we can measure the dispersion of parameter
values around any given moment value. Conceptually, this dispersion
provides us with an estimate of statistical accuracy.

More formally, later in Section \ref{subsec:front-convergence} we
establish that $\widehat{\boldsymbol{f}}$ trained using the loss
in equation (\ref{eq:front-loss-mse}) converges to $\mathbb{E}(\boldsymbol{\theta}|\boldsymbol{m})$
as $L\rightarrow\infty$. Continuing with this result, here we seek
to estimate the variance vector $\mathbb{V}\mathrm{ar}(\boldsymbol{\theta}|\boldsymbol{m})$
or the covariance matrix $\mathbb{C}\mathrm{ov}(\boldsymbol{\theta}|\boldsymbol{m})$
as a measure of statistical accuracy. We achieve this by modifying
the NNE with two changes: (i) the neural net outputs not only a point
estimate but also the variance and covariance terms, and (ii) the
loss function is specified accordingly to properly train the neural
net.

Specifically, we configure the neural net $\boldsymbol{f}$ such that
its output has two parts. The first part is a vector $\boldsymbol{\mu}$
for the point estimate and the second part specifies a covariance
matrix $\boldsymbol{V}$. We write $\boldsymbol{f}=(\boldsymbol{\mu},\boldsymbol{V})$.
In practice, we need to ensure that $\boldsymbol{V}$ is positive
definite. One way to do so is Cholesky decomposition. We let $\boldsymbol{f}$
compute the Cholesky factor of $\boldsymbol{V}$ and then transform
it into $\boldsymbol{V}$. A special case arises in applications when
researchers are interested in only the variance terms. As we will
show below, in this case we may simplify $\boldsymbol{V}$ to be a
diagonal matrix with the variance terms only. To ensure positive variances,
we can let $\boldsymbol{f}$ compute the log of the diagonal elements
and then map them into $\boldsymbol{V}$.

Using the above notation, we write $(\boldsymbol{\mu}^{(\ell)},\boldsymbol{V}^{(\ell)})=\boldsymbol{f}(\boldsymbol{m}^{(\ell)})$
for each $\ell$ in the training set. We use the following loss function
to train the neural net. This loss function is a generalization of
$C_{1}$ and known as cross-entropy loss in the machine learning literature.
\begin{align}
C_{2}(\boldsymbol{f}) & =L^{-1}\sum_{\ell=1}^{L}\left[\log\left(|\boldsymbol{V}^{(\ell)}|\right)+\left(\boldsymbol{\theta}^{(\ell)}-\boldsymbol{\mu}^{(\ell)}\right)'\left(\boldsymbol{V}^{(\ell)}\right)^{-1}\left(\boldsymbol{\theta}^{(\ell)}-\boldsymbol{\mu}^{(\ell)}\right)\right],\nonumber \\
 & \text{where }(\boldsymbol{\mu}^{(\ell)},\boldsymbol{V}^{(\ell)})=\boldsymbol{f}(\boldsymbol{m}^{(\ell)}).\label{eq:front-loss-V}
\end{align}
We train a neural net $\widehat{\boldsymbol{f}}$ that minimizes this
loss function $C_{2}$ instead of $C_{1}$.

One may note that $C_{2}$ uses the normal p.d.f. However, this does
\textit{not }mean that we must assume normality for the underlying
distribution $P(\boldsymbol{\theta}|\boldsymbol{m})$. Formally, we
will show in Section \ref{subsec:front-convergence} that $\widehat{\boldsymbol{f}}$
converges to $\mathbb{E}(\boldsymbol{\theta}|\boldsymbol{m})$ and
$\mathbb{C}\mathrm{ov}(\boldsymbol{\theta}|\boldsymbol{m})$ as $L\rightarrow\infty$.
If $\boldsymbol{V}$ is specified as diagonal, then $\widehat{\boldsymbol{f}}$
converges to $\mathbb{E}(\boldsymbol{\theta}|\boldsymbol{m})$ and
$\mathbb{V}\mathrm{ar}(\boldsymbol{\theta}|\boldsymbol{m})$ as $L\rightarrow\infty$.
Importantly, these results hold for general $P(\boldsymbol{\theta}|\boldsymbol{m})$
that may or may not be normal.\footnote{Nevertheless, it is useful to note a general result in Bayesian statistics
that posteriors tend to normal distributions with $n\rightarrow\infty$
(see \citealt{gelman2004bayes_book}, \citealt{kim2002limited}).
Thus, in applications with a large number of observations, one may
assume $P(\boldsymbol{\theta}|\boldsymbol{m})$ to be normal.}

To summarize, we use the following steps to obtain both the point
estimate and a measure of statistical accuracy. First, we follow Section
\ref{subsec:front-point} to generate the pairs $\{\boldsymbol{\theta}^{(\ell)},\boldsymbol{m}^{(\ell)}\}_{\ell=1}^{L}.$
Then, we use these pairs to train a neural net $\widehat{\boldsymbol{f}}:\boldsymbol{m}\mapsto(\boldsymbol{\mu},\boldsymbol{V})$
under the loss function $C_{2}$. Finally, we plug the moments of
the real data into $\widehat{\boldsymbol{f}}$ to obtain the point
estimate and the statistical accuracy for the point estimate.

\subsection{Convergence results\label{subsec:front-convergence}}

We provide three convergence results for NNE as well as the intuitions
for these results. The first result concerns with the neural net trained
under loss function $C_{1}$ (see equation \ref{eq:front-loss-mse}).
The second result concerns with loss function $C_{2}$ (see equation
\ref{eq:front-loss-V}). The third result concerns with a more general
loss function $C_{3}$, which we will define. All the convergence
results rely on $L\rightarrow\infty$ and hold for any finite data
size $n$. Readers more interested in applications may skip this subsection
without much loss of continuity. All proofs are given in Appendix
\ref{subsec:proofs}. 

Our results build on the existing econometric literature on the asymptotic
properties of neural nets (e.g., \citealt{chen2007sieve} and \citealt{white1990connectionist}).
The literature studies the setting where neural nets learn about the
relation from $\boldsymbol{x}_{i}$ to $\boldsymbol{y}_{i}$, such
as $\mathbb{E}(\boldsymbol{y}_{i}|\boldsymbol{x}_{i})$. That is,
neural nets are used as econometric models. This is different from
our setting where we use neural nets to recover the parameter $\boldsymbol{\theta}$
of a econometric model. In other words, the relation to be learned
is from $\{\boldsymbol{x},\boldsymbol{y}\}$ to $\boldsymbol{\theta}$.
At a high level, our propositions ``translate'' the asymptotic results
in the literature from their original setting to our setting, deriving
conditions specific to our setting so that the asymptotics can be
upheld. It is useful to note that in principle, our results are not
restricted to neural nets. Generically, we expect a machine learning
model to converge to a target relation as the training set grows,
if its capacity properly adjusts to the training set size and it is
flexible enough to include the target relation (as its capacity grows).
Our results should still apply if the neural net is replaced by a
machine learning model for which such asymptotic property is true.

We start by formalizing several notions to be used in the convergence
results. Conditional on the observed $\boldsymbol{x}$, the econometric
model together with the prior over $\Theta$ (e.g., uniform) implies
a distribution of $\boldsymbol{\theta}$ and $\boldsymbol{y}$. This
distribution in turn implies a distribution of $\boldsymbol{\theta}$
and $\boldsymbol{m}$, which we denote as $P(\boldsymbol{\theta},\boldsymbol{m})$.
Then, $\mathbb{E}(\boldsymbol{\theta}|\boldsymbol{m})$ is the best
estimate of parameter given a set of moments. In Bayesian language,
this is known as the limited-information posterior mean, with ``limited-information''
signifying that the posterior is conditional on a set of moments instead
of the entire data (\citealt{kim2002limited}).\footnote{There is a relation between our method and approximate Bayesian computation
(ABC). Conceptually, NNE can be seen as an implementation of ABC that
takes advantages of neural nets.} To express asymptotic results, we introduce a sequence of function
spaces $\{{\cal F}_{L}\}_{L=1}^{\infty}$, with each ${\cal F}_{L}$
denoting a class of neural nets. Because the exact specification of
${\cal F}_{L}$ concerns only technical aspects of proofs, we describe
it in the appendix. For our exposition below, the key property of
$\{{\cal F}_{L}\}_{L=1}^{\infty}$ is that the capacity of the neural
nets in ${\cal F}_{L}$ grows gradually with $L$. Intuitively, this
property formalizes regularization -- the neural net's capacity should
be appropriate for the size of training set. Let $\widehat{\boldsymbol{f}}_{L}$
denote a neural net in ${\cal F}_{L}$ that minimizes $C_{1}$. Let
$\lVert\cdot\rVert$ denote the 2-norm of functions: $\lVert\boldsymbol{f}\rVert\equiv\left[\int\sum_{k}f_{k}(\boldsymbol{m})^{2}dP(\boldsymbol{m})\right]^{1/2}$.
A norm of zero means $\boldsymbol{f}$ is zero a.s.

\smallskip{}

\begin{prop}
\label{prop: mean}Suppose: (i) $\Theta$ is compact, (ii) the moments
$\boldsymbol{m}$ has a compact support, (iii) $\mathbb{E}(\boldsymbol{\theta}|\boldsymbol{m})$
is continuous in $\boldsymbol{m}$. Under the loss function (\ref{eq:front-loss-mse}),
we have $\lVert\widehat{\boldsymbol{f}}_{L}-\mathbb{E}(\boldsymbol{\theta}|\boldsymbol{m})\rVert\rightarrow0$
in probability as $L\rightarrow\infty$. $\square$\smallskip{}
\end{prop}
In words, the proposition says that as we increase the number of training
datasets, NNE converges to $\mathbb{E}(\boldsymbol{\theta}|\boldsymbol{m})$
as a function of $\boldsymbol{m}$. The conditions are relatively
mild. The compactness of $\Theta$ is also assumed in standard treatment
of MLE and GMM. The support-compactness of $\boldsymbol{m}$ can be
restrictive technically but is satisfied in many applications (see
\citealt{farrel2019convergence} for some discussion). The continuity
condition requires $\mathbb{E}(\boldsymbol{\theta}|\boldsymbol{m})$
not to change abruptly when the observed data moments change only
slightly, which is satisfied in virtually any application.

We note that Proposition \ref{prop: mean} lets $L\rightarrow\infty$
while holding the data size $n$ constant. Also note that given an
application where $n$ is fixed, the accuracy of $\mathbb{E}(\boldsymbol{\theta}|\boldsymbol{m})$
as an estimate of $\boldsymbol{\theta}$ will not decrease when more
moments are added to $\boldsymbol{m}$. Instead, the accuracy will
weakly increase. Thus, with a sufficiently large $L$, the NNE will
become weakly more accurate as we add moments to $\boldsymbol{m}$.
This is in contrast to SMM, where it is known that even with the number
of simulations $R\rightarrow\infty$, redundant moments generally
lead to larger biases. We have discussed this difference in Section
\ref{subsec:front-point} from a less technical perspective.

Next, we move to loss function $C_{2}$. For the following proposition,
the neural net outputs both the point estimate and the measure of
statistical accuracy as described in Section \ref{subsec:front-accuracy}.
We use $\widehat{\boldsymbol{f}}_{L}$ to denote a neural net that
minimizes $C_{2}$ in ${\cal F}_{L}$.

\medskip{}

\begin{prop}
\label{prop: cov} Suppose conditions (i) - (iii) in Proposition \ref{prop: mean}
hold, and in addition: (iv) $\mathbb{V}\mathrm{ar}(\boldsymbol{\theta}|\boldsymbol{m})$
is continuous in $\boldsymbol{m}$ and $\mathbb{V}\mathrm{ar}(\boldsymbol{\theta}|\boldsymbol{m})\geq\delta$
for some $\delta>0$. Under loss (\ref{eq:front-loss-V}) with diagonal
covariance matrix, we have $\lVert\widehat{\boldsymbol{f}}_{L}-\boldsymbol{f}^{*}\rVert\rightarrow0$
in probability with $\boldsymbol{f}^{*}=\left[\mathbb{E}(\boldsymbol{\theta}|\boldsymbol{m}),\mathbb{V}\mathrm{ar}(\boldsymbol{\theta}|\boldsymbol{m})\right]$.\medskip{}

In addition to conditions (i)-(iv) above, suppose: (v) $\mathbb{C}\mathrm{ov}(\boldsymbol{\theta}|\boldsymbol{m})$
is continuous in $\boldsymbol{m}$ and its smallest eigenvalue is
bounded below by some positive number. Under loss (\ref{eq:front-loss-V})
with full covariance matrix, we have $\lVert\widehat{\boldsymbol{f}}_{L}-\boldsymbol{f}^{*}\rVert\rightarrow0$
in probability with $\boldsymbol{f}^{*}=\left[\mathbb{E}(\boldsymbol{\theta}|\boldsymbol{m}),\mathbb{C}\mathrm{ov}(\boldsymbol{\theta}|\boldsymbol{m})\right]$.
$\square$
\end{prop}
\medskip{}

In words, the proposition says that as we increase the number of training
datasets, NNE converges to $\mathbb{V}\mathrm{ar}(\boldsymbol{\theta}|\boldsymbol{m})$
or $\mathbb{C}\mathrm{ov}(\boldsymbol{\theta}|\boldsymbol{m})$ as
a function of $\boldsymbol{m}$, in addition to $\mathbb{E}(\boldsymbol{\theta}|\boldsymbol{m})$.
Compared to Proposition \ref{prop: mean}, the additional conditions
essentially require the underlying distribution $P(\boldsymbol{\theta}|\boldsymbol{m})$
to be bounded away from a degenerate one. Importantly, note that the
proposition does not make any distributional assumption on $P(\boldsymbol{\theta}|\boldsymbol{m})$.
In particular, it does not require $P(\boldsymbol{\theta}|\boldsymbol{m})$
to be normal. This point has been previously noted in Section \ref{subsec:front-accuracy}.

Finally, we turn to a more general loss function. The goal is to generalize
and provide some high-level intuition behind Proposition \ref{prop: mean}
and \ref{prop: cov}. To introduce this loss function, let $\phi(\cdot;\boldsymbol{\gamma})$
be a family of positive density functions for $\boldsymbol{\theta}$
parameterized by vector $\boldsymbol{\gamma}\in\Gamma$ for some space
$\Gamma$. For example, in case of the normal family, $\boldsymbol{\gamma}$
collects both the mean $\boldsymbol{\mu}$ and covariance matrix $\boldsymbol{V}$
and $\phi(\boldsymbol{\theta};\boldsymbol{\gamma})=|2\pi\boldsymbol{V}|^{-\frac{1}{2}}e^{-\frac{1}{2}(\boldsymbol{\theta}-\boldsymbol{\mu})'\boldsymbol{V}^{-1}(\boldsymbol{\theta}-\boldsymbol{\mu})}$.
Let $\boldsymbol{f}$ be a neural net that outputs the elements of
$\boldsymbol{\gamma}$. Now we may define the following loss function.

\begin{equation}
C_{3}(\boldsymbol{f})=L^{-1}\sum_{\ell=1}^{L}-\log\phi\left[\boldsymbol{\theta}^{(\ell)};\:\boldsymbol{f}(\boldsymbol{m}^{(\ell)})\right].\label{eq:front-loss-entropy}
\end{equation}
Note that $C_{3}$ reduces to $C_{2}$ if $\phi$ is the normal family.
It further reduces to $C_{1}$ if we impose an identity covariance
matrix.

To develop intuition, we consider the minimization of the limiting
version of $C_{3}$. As $L\rightarrow\infty$ , we would expect $C_{3}$
to approach $C_{3}^{*}$ defined as follows.
\[
C_{3}^{*}(\boldsymbol{f})=-\mathbb{E}\left[\log\phi\left(\boldsymbol{\theta};\:\boldsymbol{f}(\boldsymbol{m})\right)\right].
\]
The expectation on the right hand side is taken over $P(\boldsymbol{\theta},\boldsymbol{m})$.
As we know from the theory of MLE (e.g., \citealt{white1982maximum}),
minimization of $C_{3}^{*}$ amounts to the minimization of Kullback--Leibler
(KL) divergence. To be more precise, suppose that function $\boldsymbol{f}^{*}$
minimizes $C_{3}^{*}$. Then, for each possible value of $\boldsymbol{m}$,
$\boldsymbol{f}^{*}(\boldsymbol{m})$ should be equal to a value of
$\boldsymbol{\gamma}$ that makes the density $\phi(\boldsymbol{\theta};\boldsymbol{\gamma})$
closest to $P(\boldsymbol{\theta}|\boldsymbol{m})$ in terms of the
KL divergence. As it turns out, if $\phi$ is specified as the normal
family, the KL divergence is minimized when $\phi$ takes the same
mean and covariance matrix as $P(\boldsymbol{\theta}|\boldsymbol{m})$.
This result explains Proposition \ref{prop: cov}. If $\phi$ is specified
as the normal family with identity covariance matrix, then the KL
divergence is minimized when $\phi$ takes the same mean as $P(\boldsymbol{\theta}|\boldsymbol{m})$.
This result explains Proposition \ref{prop: mean}.

For subsequent applications in this paper, we will not be using loss
function $C_{3}$ -- loss functions $C_{1}$ and $C_{2}$ will be
sufficient. Nevertheless, for completeness we formalize the above
intuition with $C_{3}$ in a proposition below. For this proposition,
$\widehat{\boldsymbol{f}}_{L}$ denotes a neural net that minimizes
$C_{3}$ in ${\cal F}_{L}$. We use ${\cal KL}(\cdot\lVert\cdot)$
to denote the KL divergence.\smallskip{}

\begin{prop}
\label{prop: general} Suppose: (i) $\Theta$ is compact and $\Gamma$
is compact with a non-empty interior, (ii) $\boldsymbol{m}$ has a
compact support, (iii) $\phi$ is continuously differentiable on $\Theta\times\Gamma$,
(iv) $\mathbb{E}\left[\log\phi(\boldsymbol{\theta};\boldsymbol{\gamma})|\boldsymbol{m}\right]$
is continuous in $\boldsymbol{m}$ and $\boldsymbol{\gamma}$, (v)
for each $\boldsymbol{m}$, $\mathrm{argmin}_{\boldsymbol{\gamma}\in\Gamma}{\cal KL}\left[P(\boldsymbol{\theta}|\boldsymbol{m})\:\lVert\:\phi(\boldsymbol{\theta};\boldsymbol{\gamma})\right]$
is a single point in the interior of $\Gamma$. Let $\boldsymbol{f}^{*}(\boldsymbol{m})$
denote the minimizing point in (v). Then $\lVert\widehat{\boldsymbol{f}}_{L}-\boldsymbol{f}^{*}\rVert\rightarrow0$
in probability. $\square$\smallskip{}
\end{prop}
In words, as the number of training datasets increases, the NNE will
converge to a function that brings the parameterized family $\phi$
as close to $P(\boldsymbol{\theta}|\boldsymbol{m})$ as possible.
The key assumptions are sufficient levels of continuity and a uniqueness
condition. The uniqueness condition (i.e., assumption (v)) puts a
restriction on $\phi$, requiring that there are not multiple members
of $\phi$ being equally closest to $P(\boldsymbol{\theta}|\boldsymbol{m})$.
We note that this uniqueness condition is not imposed on the neural
net. A neural net is usually over-parameterized by a large number
of weights, and there are always multiple combinations of weights
that equally minimize the loss function. We only require $\widehat{\boldsymbol{f}}_{L}$
to take one of these weight combinations that minimize the loss function.\footnote{We thank a referee for pointing this out.}

\section{Illustration by a Simple Example\label{sec:AR1}}

In this section, we illustrate NNE through a simple AR(1) model. The
simplicity of the model permits us to graphically present NNE and
compare it to SMM. We note that for this simple econometric model,
closed-form expressions of likelihood and moments are available. Thus,
simulation-based estimators like NNE do not really offer computational
gains. However, our goal here is not to compare computational performances,
but to illustrate how NNE works and develop intuitions that extend
to more complex structural econometric models. Readers more interested
in NNE's performance in a structural application may skip this section
and go to Section \ref{sec:search}.

\subsection{Model setup and estimation\label{subsec:AR1-setup}}

We consider a simple AR(1) model. Let $\varepsilon_{i}\sim{\cal N}(0,1)$
and $\beta\in[0,1)$,
\begin{equation}
y_{i}=\beta y_{i-1}+\varepsilon_{i}.\label{eq:AR1}
\end{equation}
The model has a single parameter $\beta$. As we will see, this simple
setting makes it possible for us to graphically illustrate estimation
procedures. We assume that the initial value $y_{1}$ is drawn from
the stationary distribution, ${\cal N}(0,1/(1-\beta^{2}))$. We set
the sample size $n=100$, so that a dataset is $\{y_{i}\}_{i=1}^{100}$.
Below, we describe how NNE and SMM estimate $\beta$. Both NNE and
SMM are moment-based methods and we compare the two methods to highlight
their similarities and differences. 

\begin{algorithm}
\caption{\label{alg:NNE-AR1} NNE for AR(1)}
\onehalfspacing

\medskip{}
Suppose we have a dataset $\{y_{i}\}_{i=1}^{n}$, where $y_{i}$ follows
the AR(1) process. NNE uses the steps below to obtain an estimate
for $\beta$.
\begin{enumerate}
\item Given an index $\ell$, draw $\beta^{(\ell)}$ uniformly from an interval
such as $[0,0.9]$. Draw $y_{1}^{(\ell)}\sim{\cal N}[0,1/(1-(\beta^{(\ell)})^{2})]$.
Draw $\varepsilon_{i}^{(\ell)}\sim{\cal N}(0,1)$ and compute $y_{i}^{(\ell)}=\beta^{(\ell)}\cdot y_{i-1}^{(\ell)}+\varepsilon_{i}^{(\ell)}$
for $i=2,...,n$.
\item Compute the moment $m^{(\ell)}=\frac{1}{n-1}\sum_{i=2}^{n}y_{i}^{(\ell)}y_{i-1}^{(\ell)}$
.
\item Repeat the above for all $\ell\in\{1,...,L^{*}\}$ to obtain the set
$\{\beta^{(\ell)},m^{(\ell)}\}_{\ell=1}^{L^{*}}$. Split the set 90/10
for training and validation. Train a neural net to predict $\beta^{(\ell)}$
from $m^{(\ell)}$.
\item Apply the neural net to the moment of the original dataset, $m=\frac{1}{n-1}\sum_{i=2}^{n}y_{i}y_{i-1}$,
to obtain the estimate $\widehat{\beta}$.
\end{enumerate}
\end{algorithm}

For NNE, we list the steps in Algorithm \ref{alg:NNE-AR1}. In particular,
note here we use a single data moment, $m=\frac{1}{n-1}\sum_{i=2}^{n}y_{i}y_{i-1}$,
which is the covariance between $y_{i}$ and its own lag. This single-moment
case allows straightforward graphical presentation, which we will
show in Section \ref{subsec:AR1-compare}. We will examine using multiple
moments in Section \ref{subsec:AR1-moments}.

For SMM, we start with GMM using the same moment used by NNE above.
By equation (\ref{eq:AR1}), we have
\begin{equation}
\mathbb{E}(y_{i}y_{i-1})=\underbrace{\beta/(1-\beta^{2})}_{g(\beta)}.\label{eq:AR1-moment}
\end{equation}
We denote the right hand side as function $g(\beta)$. Because $g(\beta)$
is strictly increasing over $\beta\in[0,1)$, one could solve for
the precise value of $\beta$ if the value of $\mathbb{E}(y_{i}y_{i-1})$
were known. However, in estimation, we do not know $\mathbb{E}(y_{i}y_{i-1})$
but only a sample realization of it. This sample realization is $m=\frac{1}{n-1}\sum_{i=2}^{n}y_{i}y_{i-1}$.
GMM estimates $\beta$ by solving:
\[
g(\beta)-m=0.
\]

We know the closed-form expression for $g(\beta)$ in this simple
AR(1) model. However, in general structural estimation, closed-form
expressions of moments are often unavailable. To mirror this more
general scenario, suppose for a moment that we do not know the closed
form of $g(\beta)$ and have to evaluate it via simulations. For any
given $\beta$ we use the AR(1) model to simulate $R$ copies of the
data: $\{y_{i}^{(r)}\}_{i=1}^{n}$, $r=1,...,R$. Then, we approximate
$g(\beta)$ by $\widehat{g}(\beta)=\frac{1}{R}\sum_{r=1}^{R}(\frac{1}{n-1}\sum_{i=2}^{n}y_{i}^{(r)}y_{i-1}^{(r)})$.
A larger $R$ makes the approximation more accurate but also more
costly. SMM estimates $\beta$ by solving
\[
\widehat{g}(\beta)-m=0.
\]
Next, we illustrate NNE and SMM graphically to compare the two approaches.

\subsection{Graphic illustration of estimation\label{subsec:AR1-compare}}

Figure \ref{fig:AR1-graph} uses six plots to graphically illustrate
SMM and NNE in the estimation of AR(1). Below we discuss the plots
in order. The goal is to develop a better understanding of how NNE
approaches parameter estimation.

\begin{figure}
\begin{centering}
\begin{minipage}[t]{0.4\columnwidth}%
\begin{center}
\includegraphics[scale=0.55]{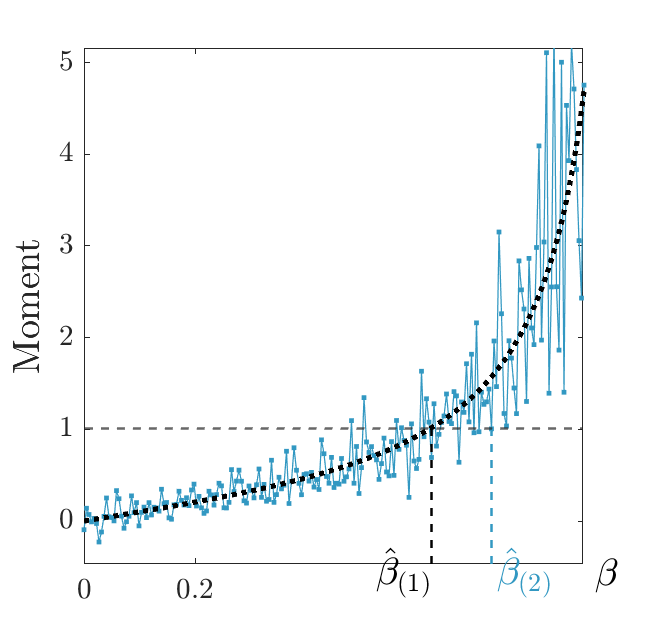}
\par\end{center}
{\footnotesize\vspace{-2em}}\textbf{(a)}{\footnotesize{} Black dotted
curve shows $\mathbb{E}(y_{t}y_{t-1})$ as a function of $\beta$.
Jagged blue curve shows the approximation by SMM with $R=1$.}{\footnotesize\par}%
\end{minipage}\hspace{5em}%
\begin{minipage}[t]{0.4\columnwidth}%
\begin{center}
\includegraphics[scale=0.55]{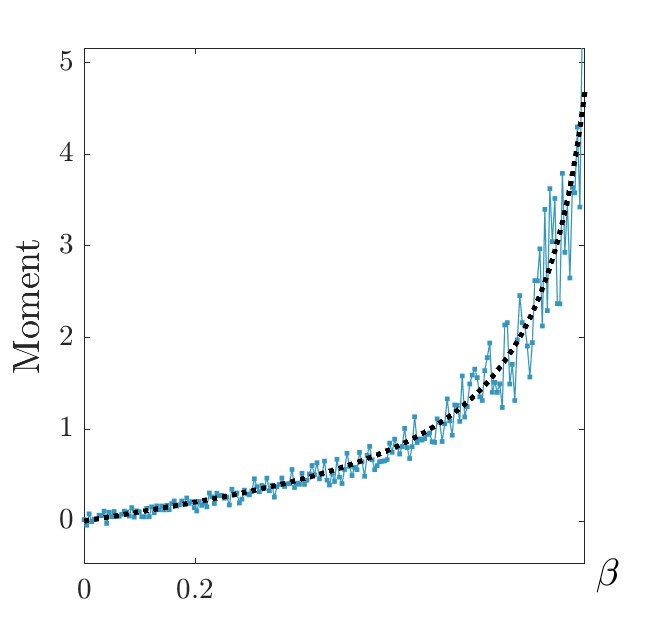}
\par\end{center}
{\footnotesize\vspace{-2em}}\textbf{(b)}{\footnotesize{} Jagged blue
curve shows the approximation by SMM with $R=5$.}{\footnotesize\par}%
\end{minipage}
\par\end{centering}
\vspace{2em}

\begin{centering}
\begin{minipage}[t]{0.4\columnwidth}%
\begin{center}
\includegraphics[scale=0.55]{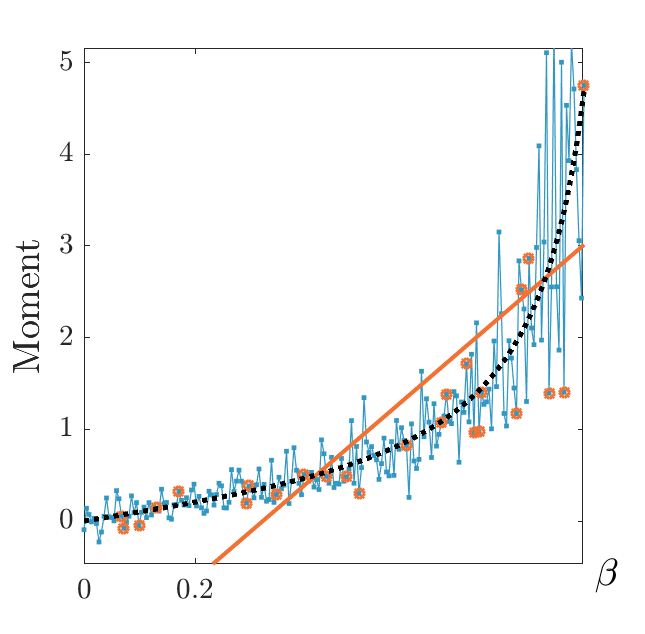}
\par\end{center}
{\footnotesize\vspace{-2em}}\textbf{(c)}{\footnotesize{} A linear fit
with $L=25$ training points sampled from jagged blue curve. These
training points are shown as red circles.}{\footnotesize\par}%
\end{minipage} \hspace{5em}%
\begin{minipage}[t]{0.4\columnwidth}%
\begin{center}
\includegraphics[scale=0.55]{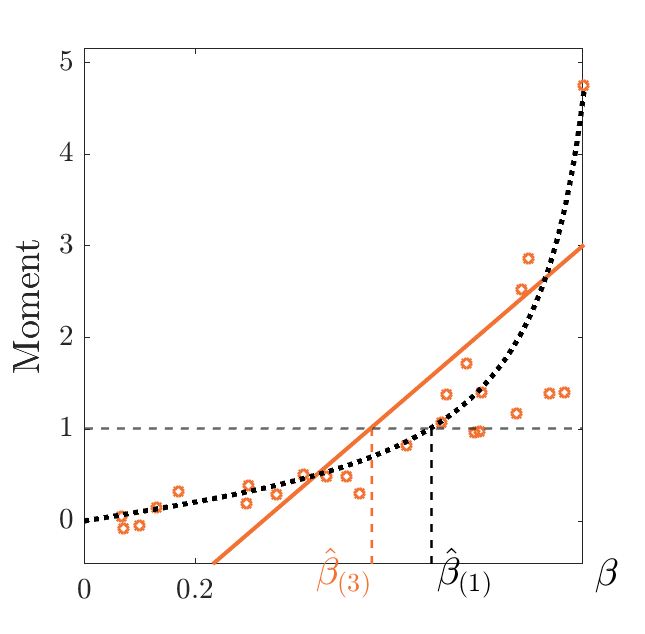}
\par\end{center}
{\footnotesize\vspace{-2em}}\textbf{(d)}{\footnotesize{} Estimation
of $\beta$ based on the linear fit.}{\footnotesize\par}%
\end{minipage}
\par\end{centering}
\vspace{2em}

\begin{centering}
\begin{minipage}[t]{0.4\columnwidth}%
\begin{center}
\includegraphics[scale=0.55]{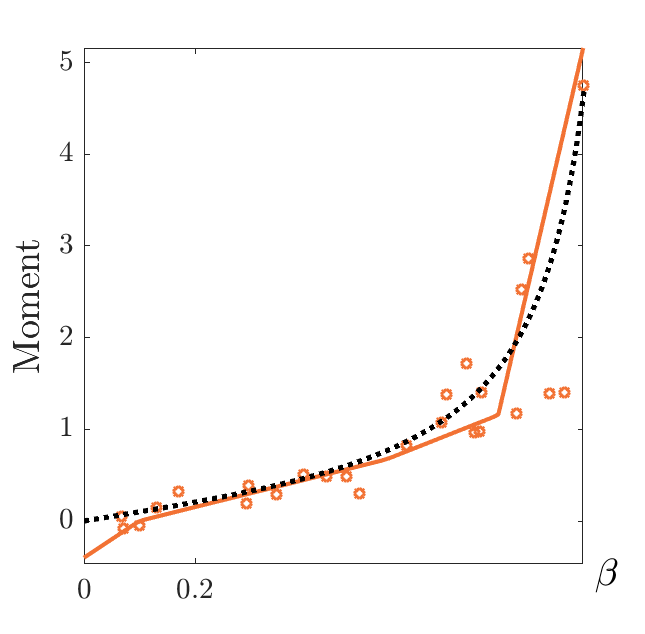}
\par\end{center}
{\footnotesize\vspace{-2em}}\textbf{(e)}{\footnotesize{} A fit by neural
net with ReLu activation, using $L=25$ training points.}{\footnotesize\par}%
\end{minipage}\hspace{5em}%
\begin{minipage}[t]{0.4\columnwidth}%
\begin{center}
\includegraphics[scale=0.55]{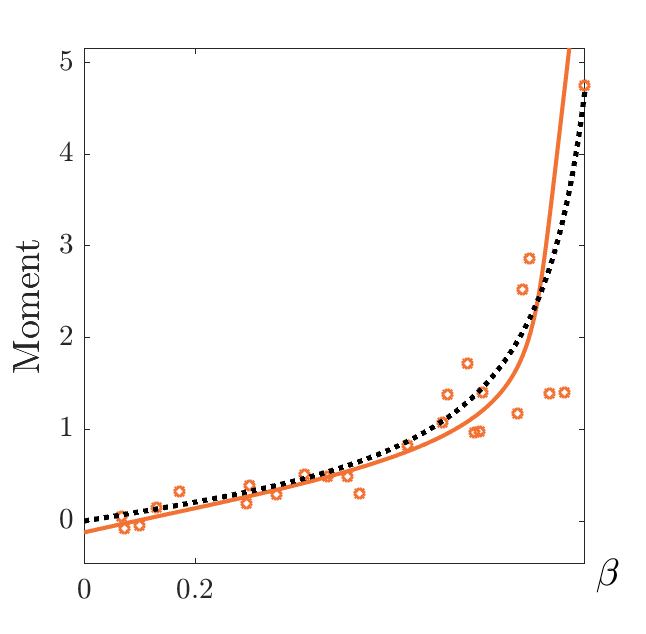}
\par\end{center}
{\footnotesize\vspace{-2em}}\textbf{(f)}{\footnotesize{} A fit by neural
net with sigmoid activation, using $L=25$ training points.}{\footnotesize\par}%
\end{minipage}
\par\end{centering}
\bigskip{}
\caption{Estimation of AR(1)\label{fig:AR1-graph}}
\end{figure}

Plot (a) and (b) illustrate SMM. The dotted black curve in plot (a)
shows $g(\beta)$ as defined in equation (\ref{eq:AR1-moment}). As
a hypothetical example, the horizontal dashed line shows the value
of $m=\frac{1}{n-1}\sum_{i=2}^{n}y_{i}y_{i-1}$ in the real data.
Under this example, GMM's estimate for $\beta$ is given by the intersection
between this horizontal dashed line and $g(\beta)$, indicated as
$\widehat{\beta}_{(1)}$ on the axis. The jagged blue curve shows
$\widehat{g}(\beta)$ evaluated with $R=1$. The SMM's estimate for
$\beta$ is given by the point on the jagged blue curve that is closest
to the horizontal dashed line, marked as $\widehat{\beta}_{(2)}$
on the axis. The gap between $\widehat{\beta}_{(1)}$ and $\widehat{\beta}_{(2)}$
is the estimation error due to simulations. This error can be reduced
by a larger $R$ at the expense of a higher computational cost. As
an example, plot (b) shows $\widehat{g}(\beta)$ evaluated with $R=5$,
which we see approximates $g(\beta)$ more closely.

Plot (c) and (d) illustrate a simplified version of NNE that replaces
neural net with a linear function. The simplification helps to demonstrate
the core idea of NNE. We use $L=25$ training points: $\{\beta^{(\ell)},m^{(\ell)}\}_{\ell=1}^{25}$.
In plot (c), these training points are shown as red circles and are
randomly sampled from the same jagged blue curve $\widehat{g}(\beta)$
as in plot (a). The red line is a linear fit of $\beta^{(\ell)}$
onto $m^{(\ell)}$. Plot (d) reproduces the horizontal dashed line
that represents a hypothetical value of $m$ in real data. The NNE's
estimate for $\beta$ is at the intersection between the red line
and horizontal dashed line, marked as $\widehat{\beta}_{(3)}$ on
the axis. The difference between $\widehat{\beta}_{(1)}$ and $\widehat{\beta}_{(3)}$
is a \textit{learning} error, caused by us using a parameterized function
(a linear fit in this case) to learn $g(\beta)$.

Plot (e) and (f) illustrate NNE. We use the same $L=25$ training
points as above. Plot (e) fits the points by a neural net with 4 hidden
neurons and ReLu activation function. We see a much better fit than
the linear function. The learning error is not marked out in the plot,
but one can see that it will be smaller than in plot (d). Aside from
ReLu, another popular activation function is the sigmoid. Plot (f)
uses a neural net with 4 hidden neurons and sigmoid activation function.
The fit is on par with that in plot (e).

Overall, plots (c) - (f) demonstrate the key idea of NNE that it attempts
a \textit{functional} approximation of the relation between moment
and parameter. Under this idea, an important detail which we have
not emphasized above is that NNE approximates the inverse function
$g^{-1}$ instead of $g$. That is, the neural net predicts $\beta^{(\ell)}$
from $m^{(\ell)}$ (not $m^{(\ell)}$ from $\beta^{(\ell)}$). The
purpose is to facilitate estimation. Specifically, suppose $\widehat{f}$
approximates $g^{-1}$, then we may estimate $\beta$ directly as
$\widehat{f}(m)$. However, suppose $\widehat{f}$ approximated $g$
instead, then we would have to solve $\widehat{f}(\beta)=m$ to obtain
an estimate of $\beta$. This problem will complicate further when
we use multiple moments because $\widehat{\boldsymbol{f}}(\beta)=\boldsymbol{m}$
typically has no exact solution.

To summarize, Figure \ref{fig:AR1-graph} shows that SMM and NNE take
different approaches and subsequently incur different types of estimation
errors. While both rely on approximations of the mapping between moment
and parameter, SMM makes a \textit{point-by-point approximation from
parameter to moment} whereas NNE makes a \textit{functional approximation
from moment to parameter}. It is not difficult to see that in cases
where the moment evaluation has sizable simulation errors (so that
the point-by-point approximation is noisy), the functional approach
by NNE can have advantages. A similar comparison applies to NNE and
indirect inference (see Appendix \ref{subsec:indirect-infer}).

\subsection{Robustness to redundant moments\label{subsec:AR1-moments}}

So far, we have focused on the estimation relying a single moment
for the purpose of visual illustration. Next, we move to estimation
using multiple moments to examine NNE's robustness to redundant moments
in the context of the AR(1) model. As discussed in Section \ref{sec:front},
it is known that redundant moments can introduce substantial biases
to SMM in finite samples. In comparison, NNE should be less affected
by redundant moments thanks to its learning mechanism.

We use Monte Carlo experiments to compare SMM and NNE with multiple
moments. The simplicity of AR(1) allows us to choose $R=\infty$ for
SMM, which is just GMM. This choice removes simulation errors and
thus helps to isolate the impact of redundant moments. For NNE, we
use $L=1000$. Our Monte Carlo exercise sets $\beta=0.6$ and generates
1000 datasets using the AR(1) model. We estimate $\beta$ from each
of these Monte Carlo datasets.

Table \ref{tab:AR1-moments} shows the results under six different
moment specifications. Row 1 is our benchmark relying on the single
moment $\frac{1}{n-1}\sum_{i=2}^{n}y_{i}y_{i-1}$. Row 2 adds the
moment $\frac{1}{n}\sum_{i=1}^{n}y_{i}^{2}$. Note that $\mathbb{E}(y_{i}^{2})=1/(1-\beta^{2})$
so the moment is informative about $\beta$. We see a slight reduction
in the RMSE for GMM, suggesting that in this case the information
gain from the added moment overcomes the potential bias increase.
We also see a slight reduction in the RMSE for NNE, showing that NNE
is able to capitalize on the information from the added moment too.

\begin{table}
\caption{Estimation of AR(1) with Different Moments\label{tab:AR1-moments}}

\begin{centering}
\begin{tabular}{lcrrcrr}
\hline 
\multicolumn{1}{l}{Moments} &  & \multicolumn{2}{c}{SMM ($R=\infty$)} &  & \multicolumn{2}{c}{NNE ($L=1e3$)}\tabularnewline
 &  & Bias & RMSE &  & Bias & RMSE\tabularnewline
\hline 
1) $y_{i}y_{i-1}$ &  & -0.019 {\footnotesize (.003)} & 0.093 {\footnotesize (.002)} &  & -0.017 {\footnotesize (.003)} & 0.091 {\footnotesize (.002)}\tabularnewline
2) $y_{i}y_{i-1}$, $y_{i}^{2}$ &  & -0.015 {\footnotesize (.003)} & 0.088 {\footnotesize (.002)} &  & -0.013 {\footnotesize (.003)} & 0.086 {\footnotesize (.002)}\tabularnewline
3) $y_{i}y_{i-k},k=1,2,3$ &  & -0.026 {\footnotesize (.003)} & 0.099 {\footnotesize (.003)} &  & -0.015 {\footnotesize (.003)} & 0.091 {\footnotesize (.002)}\tabularnewline
4) $y_{i}y_{i-k},k=1,2,...,10$ &  & -0.043 {\footnotesize (.003)} & 0.111 {\footnotesize (.003)} &  & -0.014 {\footnotesize (.003)} & 0.095{\footnotesize{} (.002)}\tabularnewline
5) $y_{i}y_{i-1}$, $y_{i}^{2}y_{i-1}$, $y_{i}y_{i-1}^{2}$ &  & -0.075 {\footnotesize (.003)} & 0.130 {\footnotesize (.003)} &  & -0.014 {\footnotesize (.003)} & 0.095 {\footnotesize (.002)}\tabularnewline
6) $y_{i}y_{i-k}$, $y_{i}^{2}y_{i-k}$, $y_{i}y_{i-k}^{2}$, $k=1,2,3$ &  & -0.086 {\footnotesize (.003)} & 0.135 {\footnotesize (.003)} &  & -0.013 {\footnotesize (.003)} & 0.096 {\footnotesize (.002)}\tabularnewline
\hline 
\end{tabular}{\footnotesize\medskip{}
}{\footnotesize\par}
\par\end{centering}
{\footnotesize Notes: SMM ($R=\infty$) is the same as GMM and follows
the usual 2-step procedure. The 1st step constructs optimal weighting
matrix accounting for serial correlations. NNE ($L=1e3$) uses a shallow
neural net of 32 hidden nodes. Reported numbers are based on 1000
Monte Carlo datasets. Numbers in parentheses are standard errors.}{\footnotesize\par}
\end{table}

\begin{figure}
\begin{minipage}[t]{0.4\columnwidth}%
\begin{center}
\includegraphics[scale=0.55]{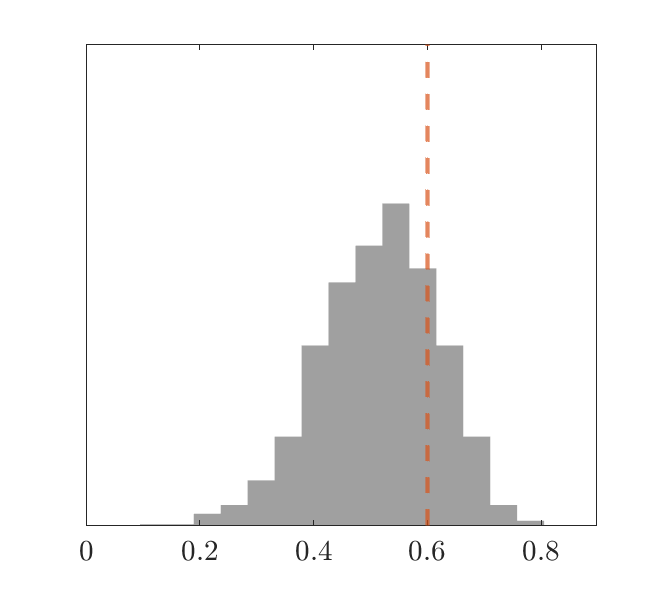}
\par\end{center}
\begin{center}
{\footnotesize\vspace{-2em}}{\small SMM with $R=\infty$}{\small\par}
\par\end{center}%
\end{minipage}\hspace{5em}%
\begin{minipage}[t]{0.4\columnwidth}%
\begin{center}
\includegraphics[scale=0.55]{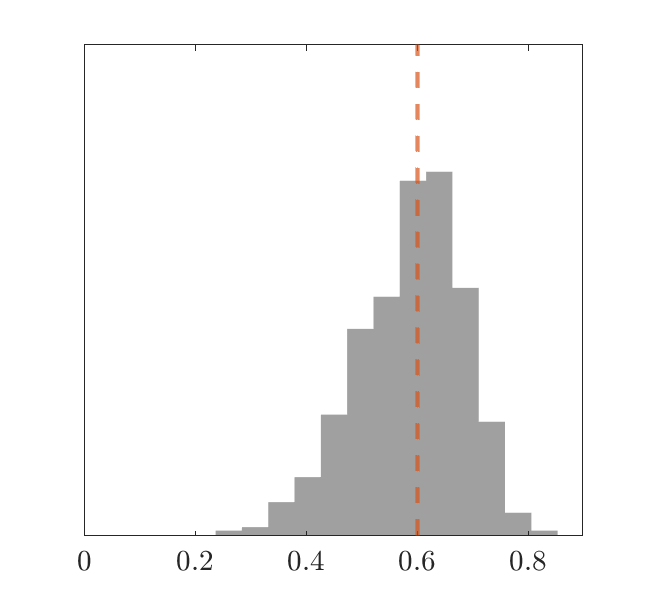}
\par\end{center}
\begin{center}
{\footnotesize\vspace{-2em}}{\small NNE with $L=1e3$}{\small\par}
\par\end{center}%
\end{minipage}\bigskip{}

{\footnotesize Notes: Histograms of the estimates of $\beta$ across
1000 Monte Carlo datasets. The vertical dashed red lines display the
true value of $\beta$. The moment specification follows row 6 in
Table \ref{tab:AR1-moments}.}{\footnotesize\par}

\medskip{}

\caption{Distributions of Estimates for AR(1) with Redundant Moments\label{fig:AR1-hist-moments}}
\end{figure}

Row 3 and 4 add moments involving further lags to the benchmark case.
Note that $\mathbb{E}(y_{i}y_{i-k})=\beta^{k}/(1-\beta^{2})$, so
each of these added moments by itself is informative about $\beta$.
However, moments with more lags (i.e., larger $k$) may become almost
redundant after we include moments with fewer lags (i.e., smaller
$k$). In both of these rows, GMM has larger bias and RMSE than the
benchmark case (row 1). In comparison, NNE's performance remains close
to the benchmark case.

Row 5 and 6 add higher-order moments to the benchmark case. While
these moments may appear informative at glance, a bit of algebra will
show that $\mathbb{E}(y_{i}^{2}y_{i-k})=0$ and $\mathbb{E}(y_{i}y_{i-k}^{2})=0$
regardless of $\beta$. Therefore, these moments are irrelevant and
redundant. In both rows, GMM now has much larger bias and RMSE than
the benchmark case. In comparison, NNE's performance remains close
to the benchmark case. For a visual comparison, Figure \ref{fig:AR1-hist-moments}
plots the histograms of the estimates using the nine moments in row
6. The bias of GMM is clear. 

To summarize, Table \ref{tab:AR1-moments} shows that NNE is less
sensitive to redundant moments, consistent with our discussion in
Section \ref{sec:front}. Note, in particular, that the simplicity
of AR(1) allows us to derive closed-form expressions of moments and
see whether a moment may be redundant. However, doing so more generally
in structural estimation is not feasible. In applications of GMM/SMM,
researchers often need to carefully select moments rather than just
``throwing in'' moments that might or might not be useful for recovering
parameters. In this regard, the robustness of NNE lessens the burden
on researchers to select moments. 

\section{Application to Consumer Sequential Search\label{sec:search}}

This section studies an application to the sequential search model
(\citealp{weitzman1979optimal}, \citealt{ursu2018power}). The model
has attracted growing interests as data on consumer online purchase
journey become increasingly available (i.e., what alternative options
consumers viewed before purchases). The model is typically estimated
by SMLE and the estimation is known to be difficult. The main challenge
comes from the large number of possible combinations of search and
purchase choices. The challenge makes likelihood evaluation difficult
and allows potential gains from using NNE.

Below, we start by briefly describing the search model and how to
estimate it with NNE and SMLE. Then, we present Monte Carlo studies
of both estimation approaches. We focus on estimation accuracy as
well as computational costs. Next, we estimate the search model using
real data. With real data, we can no longer calculate estimation accuracy
(because the true value of search model parameter is unknown). So
instead, we focus on model fit. Finally, we examine how NNE's performance
varies with the choice of moments as well as data size.

\subsection{Search model\label{subsec:search-model}}

We describe the search model in the context of hotel bookings online,
largely following \citet{ursu2018power}. A consumer is shown a list
of $J$ hotel options. She can make a search by clicking into a hotel
option. After searches, she may make a purchase by booking one of
the hotels. Let consumer $i$'s utility from booking hotel $j$ be
\[
u_{ij}=\underbrace{\boldsymbol{z}_{ij}'\boldsymbol{\beta}}_{v_{ij}}+\varepsilon_{ij},
\]
where vector $\boldsymbol{z}_{ij}$ collects hotel attributes, including
star rating, review score, location score, log price, whether it belongs
to a chain, and whether a promotion is offered. We use the log price
because price has a skewed distribution. Consumer $i$ observes $v_{ij}\equiv\boldsymbol{z}_{ij}'\boldsymbol{\beta}$
before searching. However, she observes $\varepsilon_{ij}$ only after
searching hotel $j$. An outside option where she makes no purchase
offers an utility of $u_{i0}=\eta+\varepsilon_{i0}$. This outside
utility is known to her without search. It is assumed that $\varepsilon_{ij}\sim{\cal N}(0,1)$.
The distribution of $\varepsilon_{ij}$ is assumed known to consumers
before search.

Consumers incur costs for searching. The search cost may vary by the
ranking position of a hotel in the hotel list. Let $s_{ij}\in\{1,2,...,J\}$
be the ranking of hotel $j$. We specify the cost of searching hotel
$j$ as
\[
c_{ij}=e^{\delta_{0}+\delta_{1}\log(s_{ij})}.
\]
The exponential function ensures that $c_{ij}$ is always positive.
Parameter $\delta_{1}$ captures the effect of ranking. We take the
log of $s_{ij}$ because doing so improves model fit. One convention
in the literature is to allow one ``free search,'' which we follow.
That is, consumers do not incur search cost for her first searched
hotel. The purpose is to accommodate data where observations are recorded
only if consumer has made at least one search.

The optimal search and purchase strategy is characterized in \citet{weitzman1979optimal}.
Briefly speaking, the strategy associates each product option $j$
with a ``reservation utility,'' which is a function of $v_{ij}$
and $c_{ij}$. The consumer orders her searches by reservation utility,
starting with the option with the highest reservation utility. The
consumer stops searching when the highest utility from the searched
options is higher than the maximum of the reservation utilities in
the remaining options.

In data, we observe each consumer's purchase and searches. Let $\boldsymbol{y}_{ij}$
collect two dummy variables $y_{ij}^{(\text{buy})}$ and $y_{ij}^{(\text{search})}$,
where $y_{ij}^{(\text{buy})}=1$ if option $j$ is purchased and $y_{ij}^{(\text{search})}=1$
if option $j$ is searched. Let $\boldsymbol{x}_{ij}=(\boldsymbol{z}_{ij},\log s_{ij})'$.
So the data can be denoted as $\{\{\boldsymbol{y}_{ij},\boldsymbol{x}_{ij}\}_{j=1}^{J}\}_{i=1}^{n}$.
The parameter vector $\boldsymbol{\theta}$ includes $\boldsymbol{\beta}$,
$\eta$, $\delta_{0}$, and $\delta_{1}$. A few key statistics will
be useful in our exposition: buy rate, number of searches per consumer,
and search ranking. The buy rate is the fraction of consumers who
made purchases and the search ranking is the average ranking of searched
options.

\subsection{Estimation of search model \label{subsec:search-estimation}}

We describe how NNE and SMLE estimate the consumer search model.

\subsubsection*{NNE for search model}

Algorithm \ref{alg:NNE-search} lists the main steps of NNE applied
to the consumer search model. Below, we give further details of the
algorithm (on the parameter space $\Theta$, moment choice $\boldsymbol{m}$,
and neural net training) for readers who would like to replicate our
results. Readers who are more interested in the results may skip these
details without loss of continuity.

\begin{algorithm}
\caption{\label{alg:NNE-search} NNE for consumer search}
\onehalfspacing

\medskip{}
Suppose we have a dataset $\{\{\boldsymbol{y}_{ij},\boldsymbol{x}_{ij}\}_{j=1}^{J}\}_{i=1}^{n}$,
where $\boldsymbol{y}_{ij}$ collects search and purchase choices
and $\boldsymbol{x}_{ij}$ collects hotel attributes and rankings.
We use the following steps to estimate the search model.
\begin{enumerate}
\item Given an index $\ell$, draw $\boldsymbol{\theta}^{(\ell)}$ uniformly
from the parameter space $\Theta$. Given this parameter draw and
the observed $\boldsymbol{x}_{ij}$ for all $i$ and $j$, simulate
the search and purchase outcomes $\boldsymbol{y}_{ij}^{(\ell)}$ for
all $i$ and $j$.
\item Compute the moments $\boldsymbol{m}^{(\ell)}$ that summarize $\{\{\boldsymbol{y}_{ij}^{(\ell)},\boldsymbol{x}_{ij}\}_{j=1}^{J}\}_{i=1}^{n}$.
\item Repeat the above for each $\ell\in\{1,...,L^{*}\}$ to obtain the
set $\{\boldsymbol{\theta}^{(\ell)},\boldsymbol{m}^{(\ell)}\}_{\ell=1}^{L^{*}}$.
Split the set 90/10 for training and validation. Train a neural net
to predict $\boldsymbol{\theta}^{(\ell)}$ from $\boldsymbol{m}^{(\ell)}$.
\item Apply the neural net to the moments of the original dataset to obtain
$\widehat{\boldsymbol{\theta}}$.
\end{enumerate}
\end{algorithm}

We specify the parameter space $\Theta$ as follows: $\eta\in[2,5]$,
$\delta_{0}\in[-5,-2]$, $\delta_{1}\in[-0.25,0.25]$, and $\beta_{k}\in[-0.5,0.5]$
for all $k$. This $\Theta$ covers wide ranges of key data statistics.
In a typical training set, the buy rate ranges from 0 to 1, number
of searches per consumer ranges from 1 to 23.3, and search ranking
ranges from 1.5 to 25.5. In addition, the results in this section
are not sensitive to the specification of $\Theta$. As a related
issue, in Appendix \ref{subsec:range} we examine the behavior of
NNE when $\Theta$ fails to contain the true $\boldsymbol{\theta}$
and suggest a way to check whether $\Theta$ contains the true $\boldsymbol{\theta}$.

We specify the moments $\boldsymbol{m}$ as follows. First, we include:
(i) the mean of $\boldsymbol{y}_{ij}$ and (ii) the cross-covariance
matrix between $\boldsymbol{y}_{ij}$ and $\boldsymbol{x}_{ij}$.\footnote{We do not include the covariance matrix of $\boldsymbol{y}_{ij}$
because it is implied by the mean of $\boldsymbol{y}_{ij}$. To see
it, note that $\boldsymbol{y}_{ij}$ consists of two dummy variables
and the purchase dummy can be 1 only when the search dummy is 1.} We construct additional moments by aggregating $\boldsymbol{y}_{ij}$
to the consumer level. Specifically, let $\widetilde{\boldsymbol{y}}_{i}$
collect three variables: whether consumer $i$ made any non-free search,
how many searches the consumer made, and whether the consumer made
a purchase. We add to $\boldsymbol{m}$: (iii) the mean of $\widetilde{\boldsymbol{y}}_{i}$,
(iv) the cross-covariance matrix between $\widetilde{\boldsymbol{y}}_{i}$
and $J^{-1}\sum_{j=1}^{J}\boldsymbol{x}_{ij}$, and (v) the covariance
matrix of $\widetilde{\boldsymbol{y}}_{i}$. In total, there are $2+14+3+21+6=46$
moments in $\boldsymbol{m}$. Alternative specifications of $\boldsymbol{m}$
are examined in Section \ref{subsec:search-other}.

The NNE in this section uses loss function $C_{2}$ with a diagonal
$\boldsymbol{V}$ (see Section \ref{subsec:front-accuracy}). Unless
stated otherwise, we use $L^{*}=1e4$ training and validation datasets.
The neural net uses ReLu activation function (sigmoid gives qualitatively
the same results). We use 64 hidden nodes (it is a convention to use
powers of 2). This number of hidden nodes is chosen by the validation
loss. The estimation accuracy of NNE is not sensitive to the number
of hidden nodes, so a coarse search (e.g., a few powers of 2) should
be sufficient. See Appendix \ref{subsec:config} for more details.

A final detail on the implementation of NNE is that we trim the training
and validation sets by excluding datasets where: nobody makes a purchase,
everyone makes a purchase, nobody makes non-free searches, or everyone
searches all options. These ``corner cases'' lack identification
of $\boldsymbol{\theta}$. Although trimming reduces the training
set size, we find it slightly improves the estimation accuracy of
NNE. Intuitively, this is because the training can focus on realistic
instead of corner cases.

\subsubsection*{SMLE for search model}

The prevailing approach to estimate search model is SMLE. The main
challenge of estimation lies in the enormous number (in millions/billions)
of possible search and purchase choices by a consumer. To feasibly
evaluate the likelihood over these choices, the literature applies
smoothing on the likelihood function. Yet, choosing the right amount
of smoothing is difficult (\citealt{geweke2001computationally}, \citealt{ursu2018power}).
Below, we give more details on SMLE. Readers who are more interested
in the results may skip these details without loss of continuity.

The most straightforward approach to evaluate likelihood by simulations
is the accept-reject method. Specifically, fix a consumer $i$. We
simulate $R$ draws of $\{\varepsilon_{ij}\}_{j=0}^{J}$ and count
the proportion of draws under which the model produces the same outcome
as observed in the data. A problem here is that this proportion can
often be zero. To see it, note the outcome of a consumer includes
both her search and purchase choices. Conditional on the consumer
searching $k$ options, there are $C_{k}^{J}$ possible search combinations
and $k+1$ possible purchase choices. So, in total there are $\sum_{k=1}^{J}C_{k}^{J}\cdot(k+1)$
possible outcomes for a consumer. This number is over ten billion
for $J=30$. However, a computationally feasible value of $R$ is
in the tens or hundreds. As a result, the observed outcome can often
be evaluated as a probability-zero event.

The literature uses smoothing to address the probability-zero problem.
Instead of counting only the simulations that exactly match the observed
outcome, it allows a ``partial match'' between simulations and observed
outcome. The definition of partial match is based on a set of inequalities
that characterize optimal search and purchase choices. As an example,
one inequality is that the maximum utility among searched options
exceeds the reservation utility of any unsearched option. We refer
readers to \citet{ursu2018power} for complete description of the
inequalities.\footnote{A detail here is that some inequalities require knowing the search
order. In the real data, this order is not observed and we follow
\citet{ursu2018power} to assume that higher-ranked options are searched
first. In Monte Carlo, we make the search order observable to SMLE
(which gives SMLE an advantage over NNE).} A draw of $\{\varepsilon_{ij}\}_{j=0}^{J}$ is counted as partially
matching the observed outcome if the inequalities are close to being
satisfied. A smoothing factor (or scaling factor) $\lambda$ calibrates
the partial match. A larger $\lambda$ is less lenient on unsatisfied
inequalities. $\lambda\rightarrow\infty$ is equivalent to the accept-reject
method.

A difficulty in likelihood smoothing is choosing $\lambda$. The estimate
$\widehat{\boldsymbol{\theta}}$ is often sensitive to $\lambda$.
But there is no established procedure for choosing $\lambda$ (see
discussion in \citealt{geweke2001computationally}). In practice,
researchers resort to Monte Carlo experiments to try out different
values of $\lambda$. The experiments require repeating SMLE many
times for each trial value of $\lambda$, which is computationally
expensive. Even then, it is unclear how to set $\boldsymbol{\theta}$
in these Monte Carlo experiments. One possibility is to use a preliminary
estimate of SMLE under a guess of $\lambda$. But this preliminary
estimate can be far from the true $\boldsymbol{\theta}$ due to the
sensitivity of SMLE to $\lambda$.

In this section, the optimization in SMLE starts at the center of
$\Theta$ specified above in NNE. The optimization result is not sensitive
to the starting point -- using other start points tends to lead to
almost the same result. For the number of simulations, we use $R=50$
(as in \citealt{ursu2018power}) unless stated otherwise.

\subsection{Monte Carlo studies\label{subsec:search-MC}}

In the Monte Carlo studies, we first choose a ``true'' parameter
value for the search model. Then, we use the search model to simulate
a dataset and try to recover the parameter value from the dataset.
We will refer to this simulated dataset as a ``Monte Carlo dataset,''
in order to distinguish it from the also-simulated datasets used in
training NNE.

Specifically, we generate a Monte Carlo dataset as follows. We set
$n=1000$ and $J=30$ to resemble the real data that we will estimate
later in Section \ref{subsec:search-real}. For the hotel attributes
$\boldsymbol{x}_{ij},$ we draw each attribute from a distribution
similar to its distribution in the real data.\footnote{The star rating takes 2, 3, 4, and 5 with probabilities 0.05, 0.25,
0.4, and 0.3, respectively. The review score takes 3, 3.5, 4, 4.5,
and 5 with probabilities 0.08, 0.17, 0.4, 0.3, and 0.05, respectively.
The location score follows ${\cal N}(4,0.3)$. The dummy for chain
hotels takes 1 with probability 0.8. The dummy for promotion takes
1 with probability 0.6. The log price follows ${\cal N}(0.15,0.6)$.
The ranking position $s_{ij}$ enumerates from 1 to $J=30$.} We draw $\varepsilon_{ij}$ from ${\cal N}(0,1)$. Then, we compute
$\boldsymbol{y}_{ij}$ using the optimal search and purchase strategy.
The true parameter is set at $\boldsymbol{\beta}=(0.1,\,0,\,0.2,\,-0.2,\,0.2,\,-0.2)'$,
$\eta=3$, and $\boldsymbol{\delta}=(-4,0.1)$. 

Below, we first examine the accuracy of parameter estimates and associated
computational costs. We then examine the accuracy of the economic
predictions implied by the parameter estimates.

\subsubsection*{Estimation accuracy and costs}

For either NNE or SMLE, we repeat the estimation across 100 Monte
Carlo datasets to assess the distribution of the estimates. Figure
\ref{fig:search-MC-NNE-par} plots the estimates by NNE. The red dotted
lines show the true parameter values. We see that in all plots the
estimates concentrate around the true values. Some biases may be noticeable
but are relatively small.

\begin{figure}
\begin{centering}
\hspace{-2em}\includegraphics[scale=0.85]{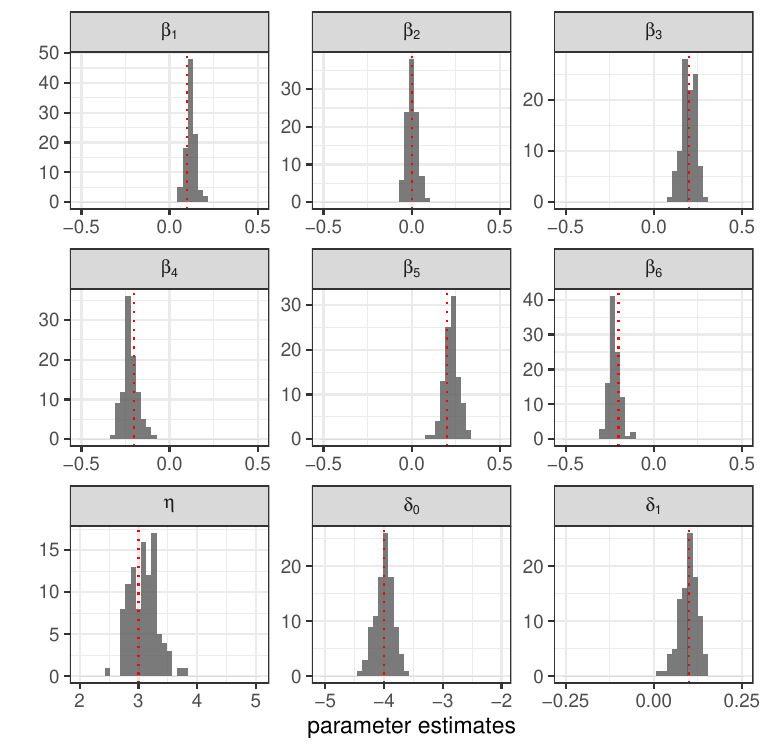}
\par\end{centering}
{\footnotesize Notes: Each plot corresponds to one parameter in the
search model. The red dotted line shows the true parameter value.
The histogram shows the parameter estimates across 100 Monte Carlo
datasets. The range of the horizontal axis equals the parameter's
range in $\Theta$. }{\footnotesize\par}

\medskip{}

\caption{\label{fig:search-MC-NNE-par} NNE's Estimates for Search Model in
Monte Carlo}
\end{figure}

\begin{figure}
\begin{centering}
\hspace{-2em}\includegraphics[scale=0.85]{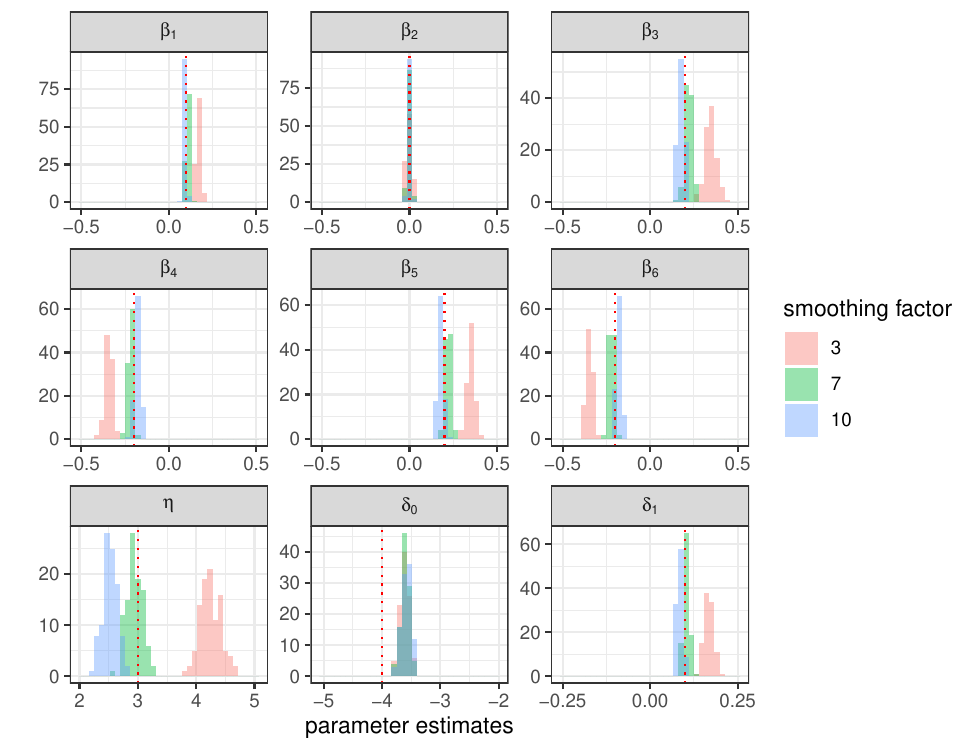}
\par\end{centering}
{\footnotesize Notes: Same as Figure \ref{fig:search-MC-NNE-par}.}{\footnotesize\par}

\medskip{}

\caption{\label{fig:search-MC-MLE-par} SMLE's Estimates for Search Model in
Monte Carlo}
\end{figure}

Figure \ref{fig:search-MC-MLE-par} plots the estimates by SMLE under
three smoothing factors: $\lambda=$ 3, 7, and 10. In particular,
we include $\lambda=7$ because it is the optimal smoothing factor
for SMLE here. Specifically, we use a grid of $\lambda$ from 1 to
15. For each point on the grid, we repeat SMLE over 100 Monte Carlo
datasets to compute a RMSE. The RMSE is smallest at $\lambda=7$.
We give the details in Appendix \ref{subsec:smoothing-factor}. It
is worth noting that obtaining this optimal $\lambda$ not only is
computationally expensive but also requires knowing the true value
of $\boldsymbol{\theta}$ (to calculate RMSE). More generally, it
is known that finding optimal $\lambda$ is difficult in the estimation
of search models.

Figure \ref{fig:search-MC-MLE-par} shows that the estimates by SMLE
are sensitive to $\lambda$. We also see that all three values of
$\lambda$ lead to some biases. As we will show in a moment, both
model fit and counterfactual analysis are more accurate with NNE than
SMLE, even when SMLE uses the optimal $\lambda=7$. It is also worth
noting that the SMLE here has a higher computational cost than the
NNE in Figure \ref{fig:search-MC-NNE-par}. We now turn to a more
detailed comparison of computational costs.

Figure \ref{fig:search-MC-rmse} plots the estimation accuracy of
SMLE and NNE as we vary computation costs (by varying $L^{*}$ and
$R$). We measure estimation accuracy by RMSE. We measure computational
costs in two ways: simulation burden and compute time. The left plot
of Figure \ref{fig:search-MC-rmse} shows the simulation burden, which
is the number of search model simulations required to carry out the
estimation (for NNE this burden equals $L^{*}$; for SMLE this burden
equals $R$ times the number of likelihood evaluations in optimization).
The simulation burden is invariant to hardware or code implementation.
The right plot of Figure \ref{fig:search-MC-rmse} shows the total
compute time, with our hardware and code.\footnote{A further computational benefit of NNE, not illustrated in Figure
\ref{fig:search-MC-rmse}, is that it is highly parallelizable. All
training datasets can be computed at once in parallel, which is useful
with computer clusters.}

\begin{figure}
\begin{centering}
\hspace{-2em}\includegraphics[scale=0.85]{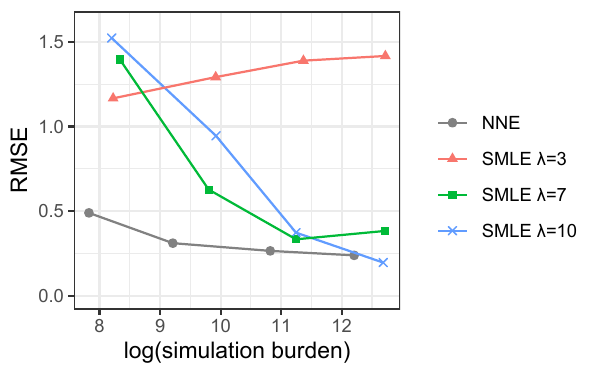}\includegraphics[scale=0.85]{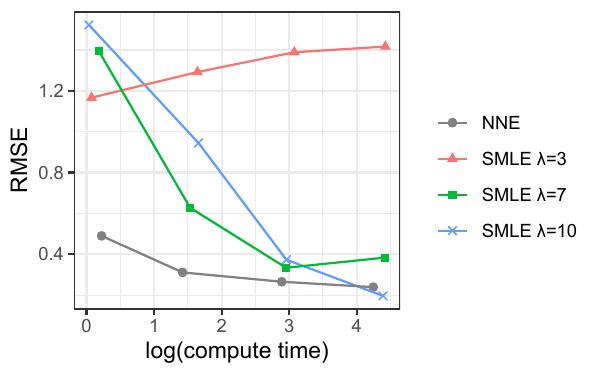}
\par\end{centering}
{\footnotesize Notes: The ``error'' in RMSE is defined as the Euclidean
distance between the true $\boldsymbol{\theta}$ and its estimate.
Simulation burden is the number of search model simulations required
to carry out estimation. Simulation burden and compute time are varied
by changing $L^{*}$ and $R$: $L^{*}\in\{2500,1e4,5e4,2e5\}$, $R\in\{5,25,100,400\}$.
The number of hidden nodes increases with $L^{*}$ and we set it to
be roughly $\propto\sqrt{L}$ (see Section \ref{subsec:front-convergence}
and Appendix \ref{subsec:proofs}).}{\footnotesize\par}

\medskip{}

\caption{\label{fig:search-MC-rmse}RMSE vs. Computational Cost for Search
Model in Monte Carlo}
\end{figure}

The results from the two plots in Figure \ref{fig:search-MC-rmse}
are consistent. NNE has better estimation accuracy than SMLE for a
wide range of computation costs, and the advantage is especially large
at low computation costs. We see that the estimation accuracy of SMLE
is sensitive to $\lambda$ and the optimal $\lambda$ changes with
computational cost (or $R$). The sensitivity to $\lambda$ makes
it important to choose the right $\lambda$. However, as discussed
before, choosing $\lambda$ is difficult (a comparison like Figure
\ref{fig:search-MC-rmse} requires repeating SMLE many times as well
as knowing the true $\boldsymbol{\theta}$).

\subsubsection*{Implications by estimates}

We examine how the parameter estimates by NNE and SMLE translate into
economic predictions about consumer behaviors. We first look at model
fit and then counterfactual analysis.

Figure \ref{fig:MC-data-pattern} shows the model fit on three key
statistics. In each plot, the red dotted line shows the average value
observed across 100 Monte Carlo datasets. The histogram shows the
values predicted by the search model estimated with these Monte Carlo
datasets. The histogram's spread captures the statistical uncertainty
in both $\widehat{\boldsymbol{\theta}}$ and the key statistic.

\begin{figure}
\begin{centering}
\begin{minipage}[t]{0.5\columnwidth}%
\begin{center}
NNE
\par\end{center}
\begin{center}
\includegraphics[scale=0.85]{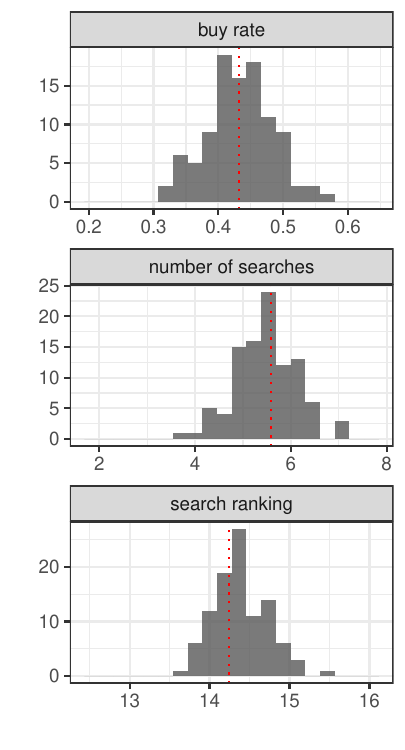}
\par\end{center}%
\end{minipage}%
\begin{minipage}[t]{0.5\columnwidth}%
\begin{center}
\hspace{-5em}SMLE \vspace{-0.5cm}
\par\end{center}
\begin{center}
\includegraphics[scale=0.85]{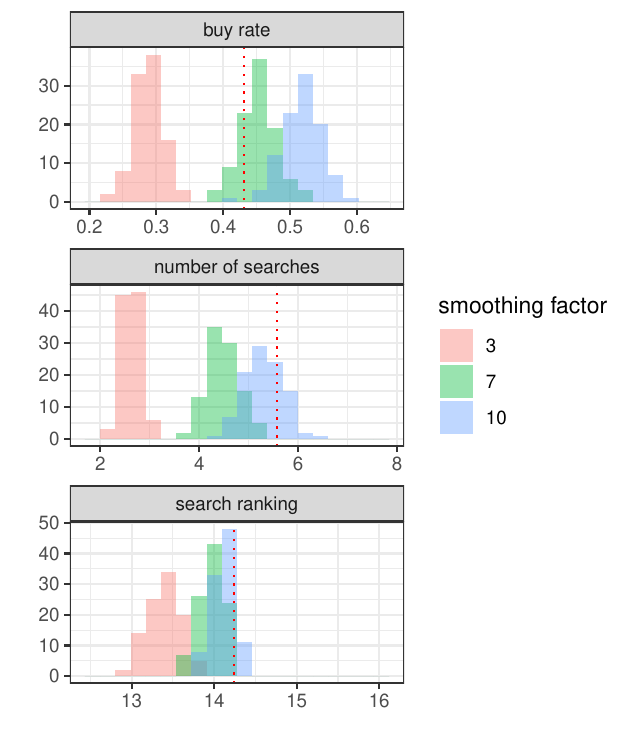}
\par\end{center}%
\end{minipage}
\par\end{centering}
{\footnotesize Notes: Plots are based on 100 Monte Carlo datasets (all
generated under the true $\boldsymbol{\theta}$). Take the 1st plot
as example. On each Monte Carlo dataset, we estimate the search model
and calculate the buy rate of a dataset simulated by the estimated
search model. The histogram plots the buy rates calculated as such
across the 100 Monte Carlo datasets. The red dotted line shows the
average of the buy rates directly observed in the Monte Carlo datasets.}{\footnotesize\par}

\medskip{}

\caption{\label{fig:MC-data-pattern}Model Fit on Key Statistics in Monte Carlo}
\end{figure}

From Figure \ref{fig:MC-data-pattern} we see that the predictions
under NNE all center around the observed averages. The SMLE again
shows sensitivity to the smoothing factor $\lambda$, and it shows
clear biases under all three smoothing factors. Biases are present
even with the optimal $\lambda=7$ obtained with expensive computations
and knowledge of true $\boldsymbol{\theta}$.

We now turn to counterfactual analysis. We focus on a counterfactual
that computes the impact of search cost on purchases. Specifically,
given a value of $\boldsymbol{\theta}$, we can use the search model
to compute the increment in buy rate if consumers no longer incur
search costs (with everything else unchanged). This increment measures
the purchases lost due to search costs. It provides the platform with
an upper bound on the effect of interventions to reduce search costs
(e.g., better web design, better ranking algorithms).

\begin{table}
\caption{\label{tab:search-mc-cf} Search Model Counterfactual -- What if
No Search Costs}

\begin{centering}
\begin{tabular}{llccccccccc}
\hline 
 &  & Truth &  & NNE &  & SMLE  &  & SMLE  &  & \multicolumn{1}{c}{SMLE }\tabularnewline
 &  &  &  &  &  & ($\lambda=3$) &  & ($\lambda=7$) &  & ($\lambda=10$)\tabularnewline
\hline 
Buy rate increase &  & 0.141 &  & 0.135 &  & 0.161 &  & 0.173 &  & \multicolumn{1}{c}{0.168}\tabularnewline
 &  & {\footnotesize (.001)} &  & {\footnotesize (.001)} &  & {\footnotesize (.001)} &  & {\footnotesize (.001)} &  & {\footnotesize (.001)}\tabularnewline
\hline 
\end{tabular}
\par\end{centering}
\medskip{}

{\footnotesize Notes: Results are averaged across 100 Monte Carlo datasets.
On each Monte Carlo dataset, we first estimate the search model and
then use the estimated search model to simulate the counterfactual.
Numbers in parentheses are standard errors for the averages.}{\footnotesize\par}
\end{table}

Table \ref{tab:search-mc-cf} displays the counterfactual results
under the true $\boldsymbol{\theta}$, NNE, and SMLE. The increment
in buy rate is closest to the truth when we use NNE's estimate. SMLE
gives us a considerably higher increment than the truth. Intuitively,
this is because search costs are over-estimated with SMLE (see Figure
\ref{fig:search-MC-MLE-par}).

\subsection{Real data estimation\label{subsec:search-real}}

We estimate the search model on a real dataset. Specifically, we use
the Expedia hotel search data in the main analysis of \citet{ursu2018power}.
There are $n=1055$ search sessions (which can be treated as individual
consumers in this search model). The number of options $J=33$ for
some consumers and $J=34$ for the others. Table \ref{tab:search-real-par}
reports the results.\footnote{Our model setup largely follows \citet{ursu2018power} but there are
a few differences, so the estimates are not directly comparable. Specifically,
we assume free first search and take the logs of ranking and price.
See Section \ref{subsec:search-model}.} We examine model fit later in Figure \ref{fig:real-data-pattern}.

\begin{table}
\caption{\label{tab:search-real-par} Search Model Estimates on Real Data}

\begin{centering}
\begin{tabular}{llrllrllrllrl}
\hline 
 &  & \multicolumn{2}{l}{NNE} &  & \multicolumn{2}{l}{SMLE} &  & \multicolumn{2}{l}{SMLE} &  & \multicolumn{2}{l}{SMLE}\tabularnewline
 &  & \multicolumn{2}{l}{} &  & \multicolumn{2}{l}{($\lambda=3$)} &  & \multicolumn{2}{l}{($\lambda=7$)} &  & \multicolumn{2}{l}{($\lambda=10$)}\tabularnewline
\hline 
Simul. burden ($\times10^{4}$) &  & \multicolumn{2}{l}{1.00} &  & \multicolumn{2}{l}{9.09} &  & \multicolumn{2}{l}{8.28} &  & \multicolumn{2}{l}{8.24}\tabularnewline
Compute time (mins) &  & \multicolumn{2}{l}{3.83} &  & \multicolumn{2}{l}{33.27} &  & \multicolumn{2}{l}{30.34} &  & \multicolumn{2}{l}{30.25}\tabularnewline
\hline 
 &  &  &  &  &  &  &  &  &  &  &  & \tabularnewline
\hline 
Stars, $\beta_{1}$ &  & 0.215 & {\footnotesize (.078)} &  & 0.218 & {\footnotesize (.109)} &  & 0.111 & {\footnotesize (.063)} &  & 0.081 & {\footnotesize (.047)}\tabularnewline
Review, $\beta_{2}$ &  & -0.089 & {\footnotesize (.107)} &  & -0.046 & {\footnotesize (.188)} &  & -0.005 & {\footnotesize (.106)} &  & 0.001 & {\footnotesize (.076)}\tabularnewline
Location, $\beta_{3}$ &  & 0.111 & {\footnotesize (.108)} &  & 0.178 & {\footnotesize (.196)} &  & 0.080 & {\footnotesize (.108)} &  & 0.056 & {\footnotesize (.077)}\tabularnewline
Chain, $\beta_{4}$ &  & 0.018 & {\footnotesize (.093)} &  & -0.012 & {\footnotesize (.119)} &  & -0.007 & {\footnotesize (.067)} &  & -0.007 & {\footnotesize (.049)}\tabularnewline
Promotion, $\beta_{5}$ &  & 0.183 & {\footnotesize (.091)} &  & 0.105 & {\footnotesize (.108)} &  & 0.050 & {\footnotesize (.062)} &  & 0.037 & {\footnotesize (.045)}\tabularnewline
Price, $\beta_{6}$ &  & -0.240 & {\footnotesize (.080)} &  & -0.390 & {\footnotesize (.117)} &  & -0.208 & {\footnotesize (.071)} &  & -0.155 & {\footnotesize (.052)}\tabularnewline
Outside, $\eta$ &  & 4.038 & {\footnotesize (.594)} &  & 3.670 & {\footnotesize (.927)} &  & 2.803 & {\footnotesize (.510)} &  & 2.557 & {\footnotesize (.361)}\tabularnewline
Search, $\delta_{0}$ &  & -2.953 & {\footnotesize (.259)} &  & -0.307 & {\footnotesize (.260)} &  & -1.290 & {\footnotesize (.191)} &  & -1.567 & {\footnotesize (.151)}\tabularnewline
Search, $\delta_{1}$ &  & 0.049 & {\footnotesize (.036)} &  & -0.147 & {\footnotesize (.093)} &  & -0.083 & {\footnotesize (.067)} &  & -0.061 & {\footnotesize (.053)}\tabularnewline
\hline 
\end{tabular}
\par\end{centering}
\medskip{}

{\footnotesize Notes: Numbers in parentheses are asymptotic standard
errors (for SMLE) and estimates of statistical accuracy outputted
by neural net (for NNE). All reported numbers are averaged across
100 estimation runs on the real data (estimates vary slightly between
estimation runs because SMLE and NNE are both simulation-based).}{\footnotesize\par}
\end{table}

The top panel of Table \ref{tab:search-real-par} reports the computational
costs, measured in terms of simulation burden and compute time. As
defined before, simulation burden is the number of search model simulations
required to carry out estimation, and compute time is total time used
to carry out estimation by our machine and code. Overall, we see the
computational cost of NNE is a number of times smaller than SMLE.

In the lower panel of Table \ref{tab:search-real-par}, the first
column reports the estimates by NNE. We see that the hotel star rating,
location score, and promotion have positive effects on the consumer
utility. The review score has a negative effect (as in \citealt{ursu2018power}),
though the estimate is not statistically significant. Price has a
negative effect and it is statistically significant. As to search
costs, the ranking effect ($\delta_{1}$) is positive but not statistically
significant. The lack of significance for $\delta_{1}$ is likely
because we allow a free search and there are only 8.8\% (or 93) consumers
who searched beyond the free search, which makes it statistically
difficult to pin down how ranking affects search cost. Overall, the
NNE's estimates are reasonable.

The next columns in Table \ref{tab:search-real-par} report the estimates
by SMLE under different smoothing factors. Consistent with what we
have seen in the Monte Carlo studies, estimates are sensitive to the
smoothing factor $\lambda$. In particular, SMLE's estimates for $\boldsymbol{\beta}$
and $\eta$ are closer to NNE's when $\lambda$ is 3 or 7. However,
SMLE's estimates for $\boldsymbol{\delta}$ are closer to NNE's when
$\lambda$ is 10. This sensitivity to $\lambda$ is problematic because,
as discussed before, choosing the optimal $\lambda$ is practically
difficult.

\begin{figure}
\begin{centering}
\begin{minipage}[t]{0.5\columnwidth}%
\begin{center}
NNE
\par\end{center}
\begin{center}
\includegraphics[scale=0.85]{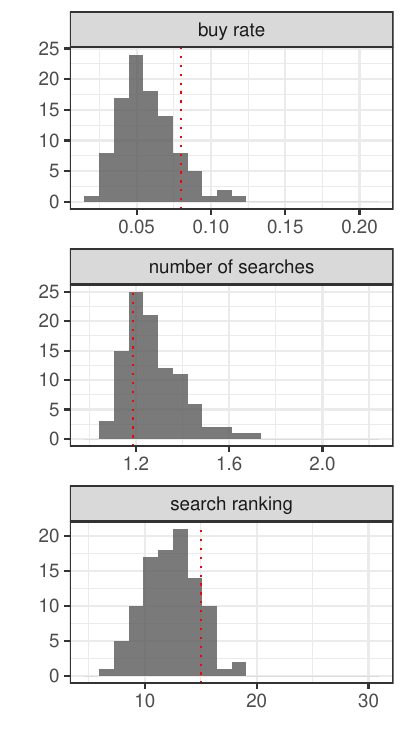}
\par\end{center}%
\end{minipage}%
\begin{minipage}[t]{0.5\columnwidth}%
\begin{center}
\hspace{-5em}SMLE \vspace{-0.5cm}
\par\end{center}
\begin{center}
\includegraphics[scale=0.85]{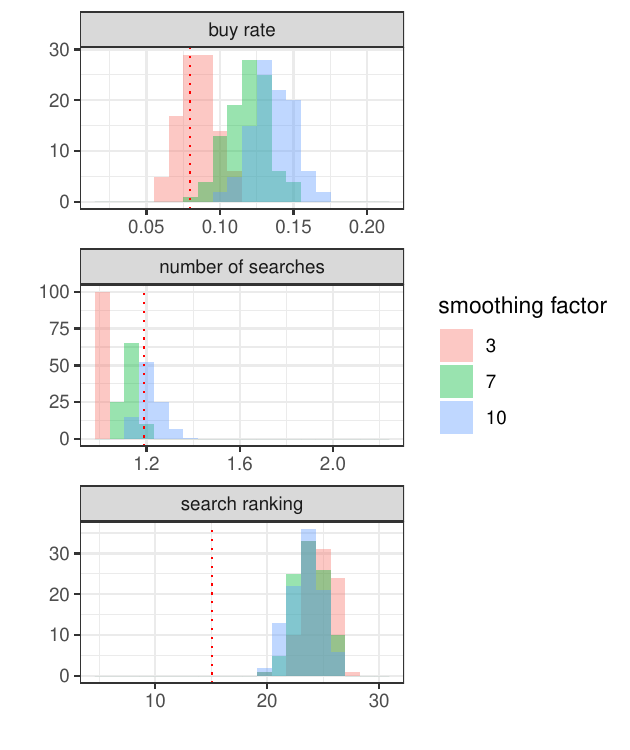}
\par\end{center}%
\end{minipage}
\par\end{centering}
{\footnotesize Notes: The histograms are based on 100 datasets bootstrapped
from the real data. Take the 1st plot as example. On each bootstrapped
dataset, we estimate the search model and calculate the buy rate of
a dataset simulated by the estimated search model. The histogram plots
the buy rates calculated as such across the 100 bootstrapped datasets.
The red dotted line shows the buy rate observed in real data.}{\footnotesize\par}

\medskip{}

\caption{\label{fig:real-data-pattern}Model Fit on Key Statistics in Real
Data}
\end{figure}

Unlike in Monte Carlo studies, we do not know the true value of the
search model parameter $\boldsymbol{\theta}$. Therefore, we cannot
calculate the RMSE of either NNE or SMLE. So instead, we focus on
model fit on three key statistics. Figure \ref{fig:real-data-pattern}
gives the results. In each plot, the red dotted line shows the value
observed in the real data. The histogram shows the predicted values
by estimated search model. The histogram's spread captures the statistical
uncertainty in both $\widehat{\boldsymbol{\theta}}$ and the key statistic.

Figure \ref{fig:real-data-pattern} shows that NNE leads to a reasonable
(albeit far from ideal) fit on the key statistics. With SMLE, the
fit varies substantially by the smoothing factor $\lambda$. This
sensitivity of SMLE to $\lambda$ is problematic because, as discussed
before, it is practically difficult to choose $\lambda$. More specifically,
SMLE can give a good fit on either the buy rate or the number of searches,
but not both. As to the search ranking, SMLE shows a substantially
worse fit than NNE. Overall, NNE seems to give a better model fit. 

The main findings from real data estimation are consistent with our
findings in Monte Carlo studies. The root of the difficulty of SMLE
lies in the enormous number of search and buy combinations. Without
smoothing, an infeasibly large number of simulations will be needed
to integrate out the unobservables for likelihood evaluation. Smoothing
makes SMLE feasible, but it also introduces biases. In contrast, NNE
does not require evaluating integrals over unobservables.

\subsection{Additional results on search model\label{subsec:search-other}}

We provide further results on the performance of NNE in estimating
the consumer search model. In particular, we examine the impacts of
moment choice and data size.

\subsubsection*{Choice of moments}

We measure how the estimation accuracy of NNE varies as the moment
choice varies. This exercise is related to the general discussion
in Section \ref{sec:front} on NNE's robustness to redundant moments.
(Also see Section \ref{sec:AR1} for a similar exercise under the
much simpler AR(1) model.)

Recall that so far we have included 46 data moments in $\boldsymbol{m}$
(see Section \ref{subsec:search-MC}). We now construct five alternative
specifications of $\boldsymbol{m}$ by removing or adding moments.
The first specification removes the covariance matrix of $\widetilde{\boldsymbol{y}}_{i}$
(40 moments remain). The second specification further removes moments
by dropping the non-free search dummy from $\widetilde{\boldsymbol{y}}_{i}$
(32 moments remain). The third specification removes $\widetilde{\boldsymbol{y}}_{i}$
all together (16 moments remain). The fourth specification adds to
the original 46 moments the cross-covariances between $\boldsymbol{y}_{ij}$
and $\boldsymbol{x}_{ij}^{2}$ (60 moments total). The fifth specification
further adds the cross-covariances between $\widetilde{\boldsymbol{y}}_{i}$
and $J^{-1}\sum_{j=1}^{J}\boldsymbol{x}_{ij}^{2}$ (81 moments total).

Table \ref{tab:search-moments} reports the estimation accuracy under
the six specifications of moments, with each row corresponding to
one specification. The first column lists the numbers of moments.
We do not see sizable increases in bias or RMSE as we move from 16
to 81 moments. Instead, the bias and RMSE largely decrease. The result
suggests that NNE is able to capitalize on the information from additional
moments while being robust to potentially redundant moments.

\begin{table}
\caption{\label{tab:search-moments}NNE in Search Model with Alternative Moment
Specifications}

\begin{centering}
\begin{tabular}{clllll}
\hline 
Number of moments &  &  & Total |bias| &  & RMSE\tabularnewline
\hline 
16 &  &  & 1.210 {\footnotesize (.027)} &  & 0.649 {\footnotesize (.017)}\tabularnewline
32 &  &  & 1.201 {\footnotesize (.026)} &  & 0.630 {\footnotesize (.015)}\tabularnewline
40 &  &  & 0.914 {\footnotesize (.029)} &  & 0.508 {\footnotesize (.015)}\tabularnewline
46 &  &  & 0.183 {\footnotesize (.029)} &  & 0.311 {\footnotesize (.017)}\tabularnewline
60 &  &  & 0.186 {\footnotesize (.035)} &  & 0.291 {\footnotesize (.015)}\tabularnewline
81 &  &  & 0.171 {\footnotesize (.033)} &  & 0.279 {\footnotesize (.015)}\tabularnewline
\hline 
\end{tabular}
\par\end{centering}
\medskip{}

{\footnotesize Notes: Each row corresponds to one specification of
moments. Total |bias| is the sum of absolute biases across the parameters
of the search model. Results are based on 100 Monte Carlo datasets.
Numbers in parentheses are standard errors.}{\footnotesize\par}
\end{table}

\subsubsection*{Data size}

Our analysis so far has focused around data size $n=1000$. We now
examine how the estimation accuracy of NNE varies with data size and
compare it to SMLE. To facilitate the comparison, we keep the simulation
burdens about the same between NNE and SMLE (by using appropriate
$L^{*}$ and $R$).

\begin{figure}
\begin{centering}
\hspace{-1em}\includegraphics[scale=0.85]{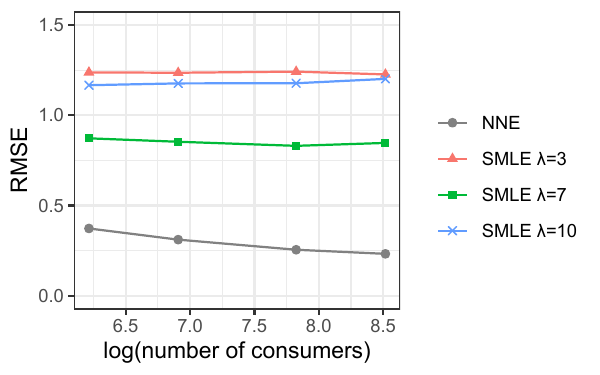}\includegraphics[scale=0.85]{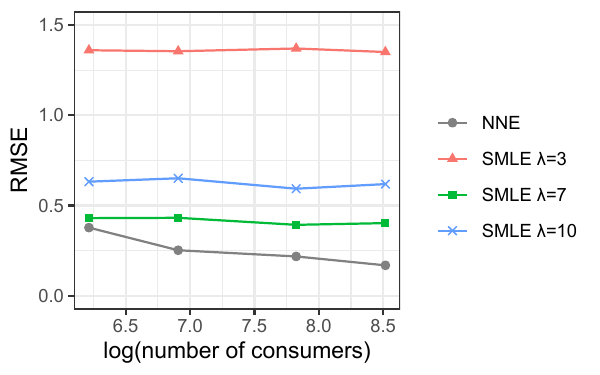}
\par\end{centering}
{\footnotesize Notes: All reported numbers are averaged across 100
Monte Carlo datasets. Within each plot, $L^{*}$ and $R$ are chosen
so that NNE and SMLE have about the same simulation burdens. In the
left plot $L^{*}=1e4$ and $R=15$. In the right plot $L^{*}=4e4$
and $R=50$.}{\footnotesize\par}

\medskip{}

\caption{\label{fig:NNE-sample-size}Estimation of Search Model with Different
Data Sizes}
\end{figure}

Figure \ref{fig:NNE-sample-size} reports the results for $n\in\{500,1000,2500,5000\}$.
The two plots differ in simulation burden, and a higher burden is
used in the right plot. Within each plot, NNE and SMLE are made to
have about the same simulation burdens. For SMLE we include three
smoothing factors: $\lambda=3,7,$ and $10$. As before, we include
$\lambda=7$ because a grid search finds it optimal for SMLE at $n=1000$
(Figure \ref{fig:NNE-sample-size} here suggests that $\lambda=7$
is also optimal for the other values of $n$). We again note that
finding this optimal $\lambda$ not only is computationally expensive
but also uses the knowledge of the true search model parameter.

The most important observation from Figure \ref{fig:NNE-sample-size}
is that the RMSE of NNE decreases with $n$ faster than SMLE. In other
words, NNE can capitalize on larger data better. Intuitively, this
is because the smoothing required by SMLE introduces bias in estimates
(as we have seen in Section \ref{subsec:search-MC}). This bias is
not a finite-sample bias and tends not to decrease with $n$. To reduce
this bias and subsequently RMSE, one needs to increase $R$ and use
a larger $\lambda$ (the optimal $\lambda$ increases with $R$).
However, a larger $R$ also implies a higher computational cost.

\section{Conclusion\label{sec:conclusion}}

We study a novel approach that leverages machine learning to estimate
the parameters of (structural) econometric models. One uses the structural
econometric model to generate datasets. These generated datasets,
together with the parameter values under which they are generated,
can be used to train a machine learning model to ``recognize'' parameter
values from datasets. As such, the approach offers a bridge to connect
the existing machine learning techniques and the current structural
econometric models.

NNE has its advantages and boundaries. We find that it has a well-defined
and meaningful limit, is robust to redundant moments, and achieves
good estimation accuracy at light computational costs in suitable
applications. The applications that can benefit the most from NNE
are where SMLE or SMM would require large numbers of simulations.
NNE is unlikely to show gains in applications where the main estimation
burden is not simulations. One example is where closed-form likelihood
functions are available. Another example is dynamic choice models
where solving the economic model constitutes the main burden in estimation.
NNE also requires that the econometric model can be used to simulate
data.

The paper leaves several possibilities unexplored, which open avenues
for future research. The first possibility is extension to very large-scale
problems where the econometric model has hundreds or more parameters,
including nuisance parameters. It is useful to note that although
in this paper NNE outputs all the parameters of the econometric model,
it is straightforward to configure NNE to output only a subset of
parameters. One can exclude nuisance parameters from this subset to
focus on the main parameters. In addition, while we have focused on
shallow neural nets for a balance between learning capacity and ease
of training, for very large-scale problems one may explore deep neural
nets.

The second unexplored possibility is to identify the relevant moments.
This exercise is useful because it can help to clarify the sources
of econometric identification and makes structural estimation more
transparent. NNE seems particularly well positioned for this exercise.
A well-trained neural net should automatically disregard irrelevant
inputs. However, a difficulty lies in the interpretability of neural
nets. The literature on interpretable machine learning may shed lights
on how to identify the inputs that a neural net disregards.

The third unexplored possibility is pre-training NNE. For any given
structural econometric model, a researcher can train a NNE and make
it available to other researchers who wish to apply the structural
model on their own data. One difficulty to overcome is that the pre-trained
NNE must not depend on a specific $\boldsymbol{x}$ (recall that in
this paper the training is conditional on $\boldsymbol{x}$). However,
if made possible, pre-trained NNE packages will significantly reduce
the costs of researchers applying structural estimation. An additional
benefit exists in handling data privacy. Because NNE only requires
data moments as inputs, with a pre-trained NNE one can carry out the
estimation without access to the full data. This feature is particularly
useful when researchers collaborate with government agencies or companies
to study data that contain sensitive personal information. Researchers
will not require access to individual-level data but only aggregated
data moments to estimate structural models.

\newpage

\appendix
\newpage

\section{Appendix}

\subsection{Number of hidden nodes\label{subsec:config}}

\begin{figure}
\begin{centering}
\includegraphics[scale=0.3]{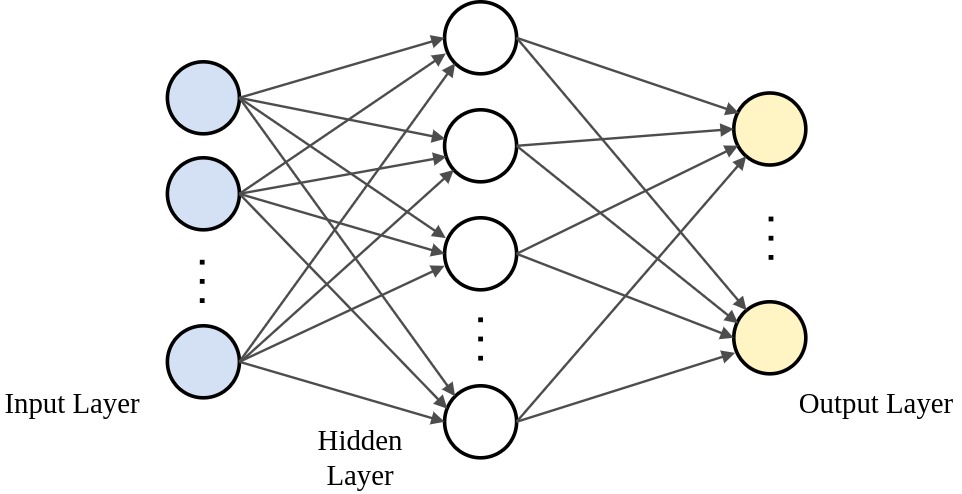}
\par\end{centering}
\bigskip{}

\caption{\label{fig:appendix-nn}A Shallow Neural Network}
\end{figure}

Figure \ref{fig:appendix-nn} visually presents a single-hidden-layer
neural net by a collection of nodes and arrows. Neural nets with a
single hidden layer are often referred to as shallow neural nets.
For shallow neural nets, an important configuration choice is the
number of nodes in the hidden layer. As we add more nodes, the neural
net becomes more flexible and capable of learning more complex mappings.
However, large neural networks may overfit the training set. The appropriate
number of hidden nodes is usually chosen by validation loss.

In the application of NNE in Section \ref{sec:search}, we use validation
loss to choose the number of hidden nodes for $L^{*}=1e4$ (which
is the setting for most of our results). We train three neural nets
with 32, 64, and 128 hidden nodes, respectively (it is a convention
to use powers of 2). The training uses examples $\ell=1,...,L$. Then,
we calculate validation losses with examples $\ell=L+1,...,L^{*}$.
The 64-node neural net typically has the smallest validation loss.
So we use 64 hidden nodes for $L^{*}=1e4$ in all results. In fact,
if averaged across 100 Monte Carlo datasets, the validation losses
($C_{2}$) with 32, 64, and 128 hidden nodes are -19.32, -20.16, and
-19.39. 

It is useful to note that in choosing the neural net configuration
via validation, all the candidate configurations use the same training
set and validation set. We do not need to simulate additional datasets
using the econometric model. So, choosing the neural net configuration
is relatively low-cost in applications where simulations of the structural
econometric model constitute the bulk of the computational cost.

A related question is how sensitive NNE's estimates are to the neural
net configuration? If they are not sensitive, one can use a coarse
grid to search for the optimal configuration. Otherwise, one needs
to do a fine search. Figure \ref{fig:network-conf} plots the estimates
by NNE under different numbers of hidden nodes. The histograms largely
overlap, which indicates that estimates are not sensitive to the number
of hidden nodes. This result suggests that a coarse grid (e.g., a
few powers of 2) should be sufficient for choosing the number of hidden
nodes.

\begin{figure}
\begin{centering}
\includegraphics[scale=0.85]{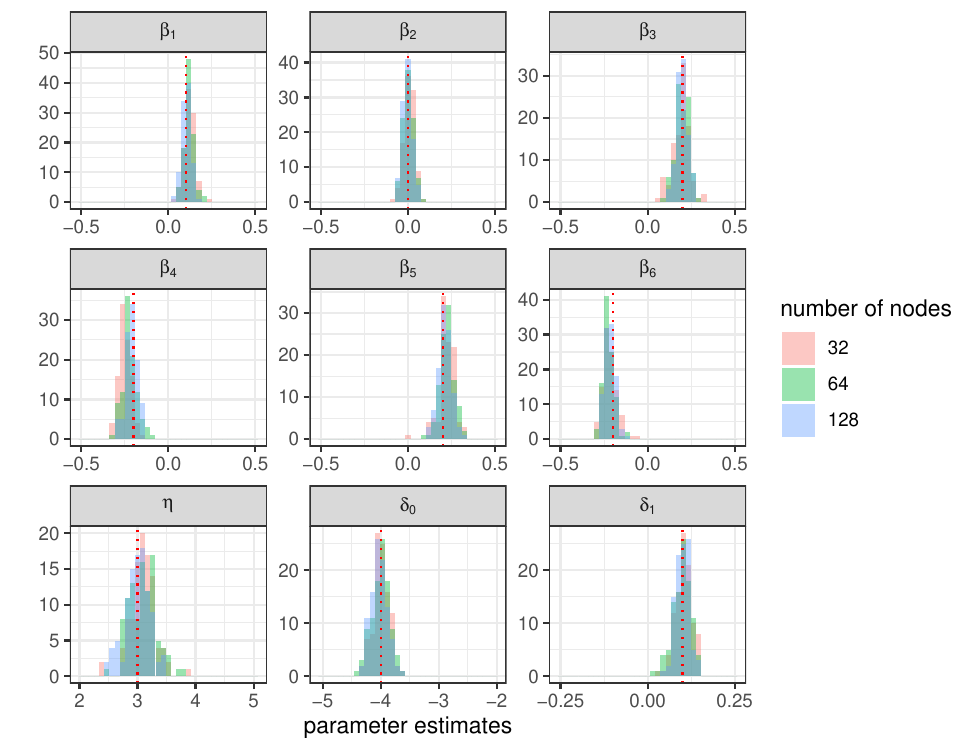}
\par\end{centering}
{\footnotesize Notes: Each plot corresponds to one parameter in the
search model. The histogram shows the parameter estimates across 100
Monte Carlo datasets. The range of the horizontal axis equals the
parameter’s range in $\Theta$.}{\footnotesize\par}

\medskip{}

\caption{\label{fig:network-conf}Estimates with Different Neural Net Configurations}

\end{figure}

\subsection{Parameter space $\Theta$\label{subsec:range}}

\begin{figure}
\begin{centering}
\includegraphics[scale=0.85]{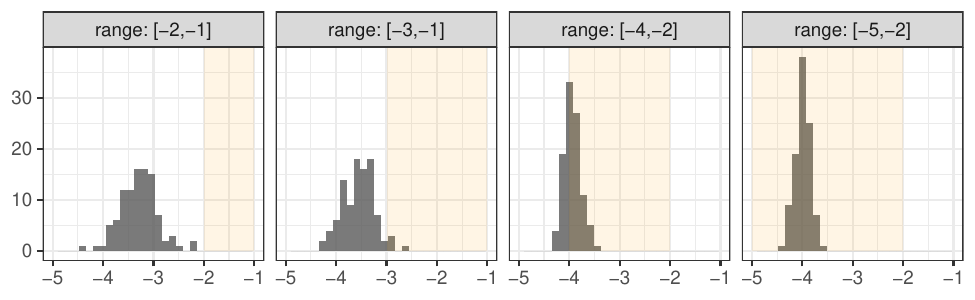}
\par\end{centering}
{\footnotesize Notes: Each histogram shows the estimates of $\delta_{0}$
by NNE across 100 Monte Carlo datasets. The four plots differ in the
specified range of $\delta_{0}$ in $\Theta$. The shaded area in
each plot shows the specified range of $\delta_{0}$. The true value
of $\delta_{0}=-4$.}{\footnotesize\par}

\medskip{}

\caption{\label{fig:par-bound}Sensitivity Analysis of NNE to $\Theta$}
\end{figure}

In NNE, the parameter values of the econometric model in training
and validation sets are drawn from a parameter space $\Theta$. A
practical question is how NNE will behave if $\Theta$ fails to contain
the true value of $\boldsymbol{\theta}$. We explore this question
in consumer search model (Section \ref{sec:search}). We will focus
on the search cost parameter $\delta_{0}$, but results with other
parameters are similar. Recall that we have specified the range for
$\delta_{0}$ as $[-5,-2]$ and the true value of $\delta_{0}=-4$.

Figure \ref{fig:par-bound} shows the distributions of NNE's estimates
for $\delta_{0}$ under different specified ranges of $\delta_{0}$.
The ranges in the first two plots do not include the truth $\delta_{0}=-4$.
In both plots, most of the estimates fall between the specified range
and the true value. The range in the third plot has the truth at its
boundary. The range in the last plot contains the truth in its interior.
In these two plots, the estimates closely center around the truth.

Overall, we see that NNE is well-behaved even when $\Theta$ is mis-specified.
The estimate typically falls outside $\Theta$ in the direction towards
the true parameter value. This behavior suggests we check whether
NNE's estimate is inside $\Theta$. If not, then $\Theta$ likely
does not contain the truth and needs to be adjusted.

\subsection{Optimal smoothing factor for SMLE\label{subsec:smoothing-factor}}

\begin{figure}
\begin{centering}
\hspace{-2em}\includegraphics[scale=0.85]{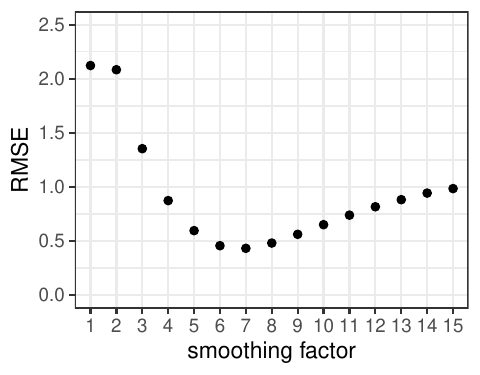}
\par\end{centering}
{\footnotesize Notes: For each value of the smoothing factor, we repeat
SMLE across 100 Monte Carlo datasets to calculate the RMSE. SMLE here
uses $R=50$. The ``error'' in RMSE is defined as the Euclidean
distance between the true value and the estimate of $\boldsymbol{\theta}$.}{\footnotesize\par}

\medskip{}

\caption{\label{fig:MLE-smoothing-RMSE}Optimal Smoothing Factor in SMLE}
\end{figure}

To estimate the consumer search model with SMLE requires a smoothing
factor. We use grid search to find the optimal smoothing factor $\lambda$
for $R=50$ (which is the setting for most of our results). Specifically,
we use a grid of $\lambda$ from 1 to 15. For each value on the grid,
we repeat SMLE across 100 Monte Carlo datasets to calculate the RMSE
(so the SMLE is repeated 1500 times in total). Figure \ref{fig:MLE-smoothing-RMSE}
reports the results. We see $\lambda=7$ achieves the lowest RMSE.

\subsection{Estimates of statistical accuracy}

\begin{figure}
\begin{minipage}[t]{0.45\columnwidth}%
\begin{center}
NNE \vspace{-0.25cm}
\par\end{center}
\begin{center}
\includegraphics[scale=0.85]{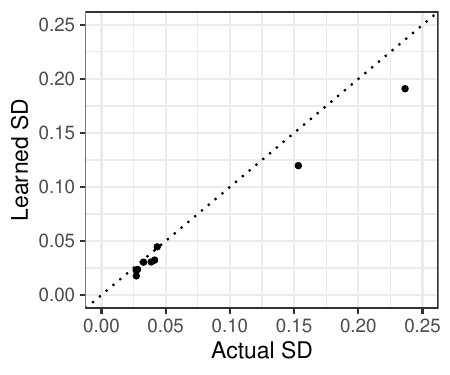}
\par\end{center}%
\end{minipage}\hfill{}%
\begin{minipage}[t]{0.45\columnwidth}%
\begin{center}
SMLE \vspace{-0.25cm}
\par\end{center}
\begin{center}
\includegraphics[scale=0.85]{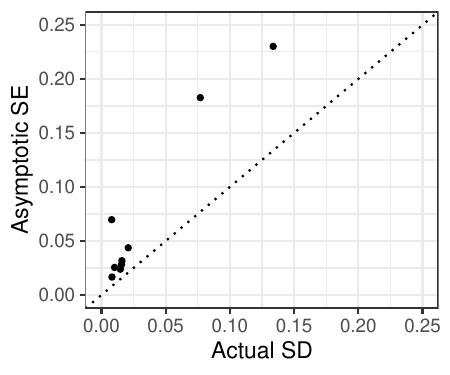}
\par\end{center}%
\end{minipage}

{\footnotesize Notes: Each point represents one parameter in the search
model.}{\footnotesize\textcolor{red}{{} }}{\footnotesize Horizontal axis
shows the standard deviation of the point estimates across 100 Monte
Carlo datasets. For NNE (left plot), vertical axis shows the estimate
of statistical accuracy outputted by neural net, averaged across Monte
Carlo datasets. For SMLE (right plot), vertical axis shows the asymptotic
standard error, averaged across Monte Carlo datasets.}{\footnotesize\par}

\medskip{}

\caption{\label{fig:search-MC-se}Estimates of Statistical Accuracy}
\end{figure}

Section \ref{subsec:front-accuracy} describes how NNE can output
an estimate of statistical accuracy in addition to the point estimate.
Section \ref{subsec:front-convergence} provides the theoretical properties
of this estimate of statistical accuracy. Below, we examine this estimate
of statistical accuracy in the consumer search model, by comparing
it to the standard deviation of the point estimates across Monte Carlo
datasets.

The left plot of Figure \ref{fig:search-MC-se} shows the results.
Each point represents one parameter in the search model (there are
9 in total). The horizontal axis shows the standard deviation of the
point estimates. The vertical axis shows the estimate of statistical
accuracy given by NNE. Recall that in the search model we have used
the loss function $C_{2}$ with a diagonal $\boldsymbol{V}$. So the
estimates of statistical accuracy are taken as square roots of the
diagonal entries. We see that the points are not far from the $45^{\circ}$
line, indicating that NNE estimates the statistical accuracy reasonably
well. 

As a comparison, the right plot of Figure \ref{fig:search-MC-se}
repeats the above exercise but for SMLE ($\lambda=7$). The vertical
axis shows the asymptotic standard error. We see larger deviations
from the $45^{\circ}$ line. The deviations are likely caused by smoothing.
Smoothing affects the likelihood function's slope and curvature, which
the asymptotic formula relies on. Thus, it is not surprising that
smoothing can distort the asymptotic standard errors of SMLE.

\subsection{Regularized polynomial instead of neural net\label{subsec:lasso}}

We examine how the estimation accuracy of NNE is affected if the neural
net is replaced by polynomial regression. Specifically, we train a
(multivariate) polynomial from $\boldsymbol{m}$ to the econometric
model's parameter. For example, if $\boldsymbol{m}$ collects 10 moments,
a second-degree polynomial will have 65 terms and a third-degree polynomial
will have 285 terms. The polynomial is estimated using lasso. The
regularization factor is chosen using the validation set ($\ell=L+1,...,L^{*}$).

Table \ref{tab:AR1-lasso} replicates Table \ref{tab:AR1-moments}
but with lasso polynomials. There are two observations when we compare
the numbers in Table \ref{tab:AR1-moments} to Table \ref{tab:AR1-lasso}.
The first observation is that the lasso shows larger biases than the
neural net (while not reported in the table, the biases with fourth-degree
polynomial are still about twice the biases with the neural net).
This observation indicates that the polynomials are not as flexible
to capture the relation between $\boldsymbol{m}$ and $\beta$. The
second observation is that unlike SMM, lasso does not show clear increases
in bias as we add moments. So lasso exhibits some robustness to redundant
moments that we have seen with neural nets. This observation is intuitive
because both the neural net and lasso can use the training set to
learn which moments contribute to estimation.

\begin{table}
\caption{Estimation of AR(1) with Different Moments, Lasso\label{tab:AR1-lasso}}

\begin{centering}
\begin{tabular}{lcccccc}
\hline 
\multicolumn{1}{l}{Moments} &  & \multicolumn{2}{c}{Lasso degree 2 ($L$=$1e3$)} &  & \multicolumn{2}{c}{Lasso degree 3 ($L$=$1e3$)}\tabularnewline
 &  & Bias & RMSE &  & Bias & RMSE\tabularnewline
\hline 
1) $y_{i}y_{i-1}$ &  & -0.102 {\footnotesize (.003)} & 0.132 {\footnotesize (.002)} &  & -0.055 {\footnotesize (.003)} & 0.114 {\footnotesize (.002)}\tabularnewline
2) $y_{i}y_{i-1}$, $y_{i}^{2}$ &  & -0.067 {\footnotesize (.002)} & 0.101 {\footnotesize (.002)} &  & -0.046 {\footnotesize (.003)} & 0.100 {\footnotesize (.002)}\tabularnewline
3) $y_{i}y_{i-k},k=1,2,3$ &  & -0.066 {\footnotesize (.003)} & 0.114 {\footnotesize (.002)} &  & -0.044 {\footnotesize (.003)} & 0.106 {\footnotesize (.002)}\tabularnewline
4) $y_{i}y_{i-k},k=1,2,...,10$ &  & -0.055 {\footnotesize (.003)} & 0.112{\footnotesize{} (.002)} &  & -0.040 {\footnotesize (.003)} & 0.107{\footnotesize{} (.002)}\tabularnewline
5) $y_{i}y_{i-1}$, $y_{i}^{2}y_{i-1}$, $y_{i}y_{i-1}^{2}$ &  & -0.095 {\footnotesize (.003)} & 0.130 {\footnotesize (.002)} &  & -0.048 {\footnotesize (.003)} & 0.111 {\footnotesize (.002)}\tabularnewline
6) $y_{i}y_{i-k}$, $y_{i}^{2}y_{i-k}$, $y_{i}y_{i-k}^{2}$, $k=1,2,3$ &  & -0.065 {\footnotesize (.003)} & 0.115 {\footnotesize (.002)} &  & -0.043 {\footnotesize (.003)} & 0.108 {\footnotesize (.002)}\tabularnewline
\hline 
\end{tabular}\medskip{}
\par\end{centering}
{\footnotesize Notes: This table replicates Table \ref{tab:AR1-moments},
except that polynomials replace the neural net in NNE. Results are
based on 1000 Monte Carlo datasets. Numbers in parentheses are standard
errors.}{\footnotesize\par}
\end{table}

\subsection{Indirect inference\label{subsec:indirect-infer}}

A graphical illustration as in Figure \ref{fig:AR1-graph} can be
made for indirect inference. Indirect inference extends SMM by allowing
one to use the parameter of an auxiliary econometric model as moments.
We note AR(1) is a simple model that does not require indirect inference
to estimate. Nevertheless, the setting works well for illustrating
the conceptual difference between indirect inference and NNE. 

The auxiliary model is typically mis-specified. We use MA(1) as the
auxiliary model: $y_{i}=\varepsilon_{i}+\alpha\varepsilon_{i-1}$
with $\varepsilon_{i}\sim{\cal N}(0,1)$. There are two common estimators
for MA(1). The first estimator uses the auto-covariance in MA(1):
$\mathbb{E}(y_{i}y_{i-1})=\alpha$, which suggests estimating $\alpha$
with $\widehat{\alpha}^{AC}=\frac{1}{n-1}\sum_{i=2}^{n}y_{i}y_{i-1}$.
The superscript ``AC'' signifies auto-covariance. Note that $\widehat{\alpha}^{AC}$
coincides with the moment used by SMM in Section \ref{subsec:AR1-setup}.
Thus, the consequent indirect inference estimator for $\beta$ is
effectively the same as the SMM illustrated in Figure \ref{fig:AR1-graph}(a)
in Section \ref{sec:AR1}.

\begin{figure}
\begin{centering}
\includegraphics[scale=0.58]{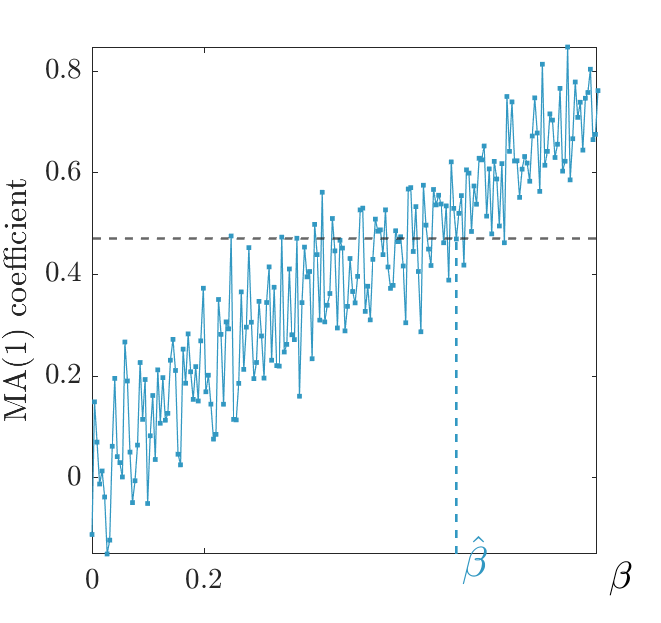}
\par\end{centering}
{\footnotesize Notes: Jagged blue curve shows the least-square estimate
of the MA(1) coefficient as a function of $\beta$.}{\footnotesize\par}

\medskip{}

\caption{\label{fig:indirect_infer}Indirect Inference for AR(1) }
\end{figure}

The second estimator for MA(1) uses least squares. It typically specifies
$\varepsilon_{0}=0$ and then uses the MA(1) to compute $\varepsilon_{1}=y_{1}$,
$\varepsilon_{2}=y_{2}-\alpha\varepsilon_{1}$, $\varepsilon_{3}=y_{3}-\alpha\varepsilon_{2}$,
... One estimates $\alpha$ with $\widehat{\alpha}^{LS}=\text{argmin}_{\alpha}\sum_{i=1}^{n}\varepsilon_{i}^{2}$.
The consequent indirect inference estimator is illustrated in Figure
\ref{fig:indirect_infer}. The jagged curve shows, for each $\beta$,
the value of $\widehat{\alpha}^{LS}$ estimated from a dataset simulated
with AR(1) under that $\beta$. Suppose that the horizontal dashed
line shows the value of $\widehat{\alpha}^{LS}$ from real data. Then,
indirect inference estimates $\beta$ by the point on the jagged curve
closest to the horizontal dashed line, marked as $\widehat{\beta}$
on the horizontal axis. 

The discussion above shows a conceptual similarity between indirect
inference and SMM that they both attempt approximations \textit{point-by-point}
over $\beta$. In contrast, NNE attempts to learn a \textit{functional}
approximation (discussed in Section \ref{sec:AR1}).

\newpage

\section{Proofs\label{subsec:proofs}}

We first set up the preliminaries for proving Proposition \ref{prop: mean}-\ref{prop: general}.
Our asymptotics ($L\rightarrow\infty$) are stated for any fixed data
size $n$. Conditional on the observed $\boldsymbol{x}$, let $P$
denote the probability distribution for $(\boldsymbol{\theta},\boldsymbol{m})$,
where $\boldsymbol{\theta}\sim{\cal U}(\Theta)$ and $\boldsymbol{m}|\boldsymbol{\theta}$
follows the distribution implied by the econometric model $\boldsymbol{q}$
and the specification of $\boldsymbol{m}$ as a function of $\{\boldsymbol{y},\boldsymbol{x}\}$.
Then, the training examples $\{\boldsymbol{\theta}^{(\ell)},\boldsymbol{m}^{(\ell)}\}_{\ell=1}^{L}$
are i.i.d. samples from $P$.

We now work towards a lemma that states the general conditions for
a sequence of neural networks to converge to a general target function.
This lemma forms the basis of all our proofs. We first define the
function space where the target function resides. Let ${\cal M}$
denote the support for $\boldsymbol{m}$ under $P$. Let ${\cal F}$
denote the function space ${\cal F}\equiv\{\boldsymbol{f}:{\cal M}\rightarrow\Delta,\boldsymbol{f}\text{ is continuous}\}$,
where $\Delta$ is some subset of a Euclidean space. Let $\lVert\cdot\rVert$
denote the 2-norm, that is, $\lVert\boldsymbol{f}\rVert^{2}=\int\lVert\boldsymbol{f}(\boldsymbol{m})\rVert^{2}dP(\boldsymbol{m})$.
All our propositions share the common goal of learning some target
function $\boldsymbol{f}^{*}\in{\cal F}$ that maps ${\cal M}$ to
the interior of $\Delta$. For example, $\boldsymbol{f}^{*}$ is $\mathbb{E}(\boldsymbol{\theta}|\boldsymbol{m})$
in Proposition \ref{prop: mean}.

Next, we need to define the sequence of trained neural networks that
we want to converge to the target function as $L\rightarrow\infty$.
A proper definition of a trained neural network requires two elements,
the loss function and the class of neural nets within which we minimize
the loss function. For our purposes, it is sufficiently general to
consider loss function constructed from an individual loss function
$h:\Theta\times\Delta\rightarrow\mathbb{R}$:
\[
C(\boldsymbol{f})=L^{-1}\sum_{\ell=1}^{L}h\left[\boldsymbol{\theta}^{(\ell)},\boldsymbol{f}(\boldsymbol{m}^{(\ell)})\right].
\]
We are interested in $\widehat{\boldsymbol{f}}_{L}$ such that $C(\widehat{\boldsymbol{f}}_{L})-\mathrm{inf}_{\boldsymbol{f}\in{\cal F}_{L}}C(\boldsymbol{f})$
is a term that converges to 0 in probability as $L\rightarrow\infty$.
If the infimum can be attained, then we can simply write $\widehat{\boldsymbol{f}}_{L}\in\mathrm{argmin}_{\boldsymbol{f}\in{\cal F}_{L}}C(\boldsymbol{f})$.
Here, ${\cal F}_{L}$ is the class of neural nets in which we minimize
the loss function.

We focus on the sequence ${\cal F}_{L}$ constructed from shallow
neural nets, which are sufficient for our purpose. Asymptotic properties
of shallow neural nets have been well studied in econometric literature
(deep neural nets are very recently studied by \citealt{farrel2019convergence}).
We adopt the shallow neural nets used in \citet{white1990connectionist}.
Let $\psi:\mathbb{R}\rightarrow\mathbb{R}$ be any sigmoid function
such as the logistic function. Let $v$ denote the dimension of $\boldsymbol{m}$.
Let
\begin{align}
{\cal F}(u,b) & \equiv\left\{ \boldsymbol{f}\in{\cal F}:f_{k}(\boldsymbol{m})=w_{1k0}+\sum_{j=1}^{u}w_{1kj}\cdot\psi\Big(w_{0j0}+\sum_{i=1}^{v}w_{0ji}\cdot m_{i}\Big),\right.\nonumber \\
 & \left.\hspace{10em}\sum_{j=0}^{u}|w_{1kj}|\leq b,\;\sum_{j=1}^{u}\sum_{i=0}^{v}|w_{0ji}|\leq u\cdot b\right\} .\label{eq:F(q,B)}
\end{align}
In words, ${\cal F}(u,b)$ denotes the single-hidden-layer neural
networks with $u$ hidden units and weights $\boldsymbol{w}$ bounded
in a way by $b$. A finite $b$ makes ${\cal F}(u,b)$ a compact space,
which allows the infimum of the loss function to be attainable. We
relax this bound as $L$ increases. Specifically, our lemma holds
for any sequence $\{{\cal F}_{L}\}_{L=1}^{\infty}$ where ${\cal F}_{L}={\cal F}(u_{L},b_{L})$
with $u_{L}$ and $b_{L}$ that grow sufficiently slow. By \citet{white1989learning},
a permissible choice is $u_{L}\propto\sqrt{L}$ and $b_{L}\propto\log(L)$.

We are in a position to state the lemma that links $\widehat{\boldsymbol{f}}_{L}$
to $\boldsymbol{f}^{*}$.

\medskip{}

\begin{lem}
\label{lem:sieve}Let ${\cal M}$ and $\Theta$ be compact and let
$\Delta$ be compact with a non-empty interior. In addition, assume
the following:\vspace{0.5em}

(i) For any $\epsilon>0$, $\mathbb{E}\left[h(\boldsymbol{\theta},\boldsymbol{f}^{*}(\boldsymbol{m}))\right]<\inf_{\boldsymbol{f}\in{\cal F}:\lVert\boldsymbol{f}-\boldsymbol{f}^{*}\rVert\geq\epsilon}\mathbb{E}\left[h(\boldsymbol{\theta},\boldsymbol{f}(\boldsymbol{m}))\right]$.

(ii) $h$ is continuously differentiable over $\Theta\times\Delta$.

\vspace{0.5em}Then, we have $\lVert\widehat{\boldsymbol{f}}_{L}-\boldsymbol{f}^{*}\rVert\rightarrow0$
in probability. $\square$
\end{lem}
\begin{proof}
of \textbf{Lemma} \ref{lem:sieve}: We prove the lemma using \citet{chen2007sieve}.
Specifically, we need to show the following conditions are satisfied.
Recall that $\lVert\cdot\rVert$, when applied to a function, denotes
the 2-norm on ${\cal F}$.\medskip{}

1. $\mathbb{E}\left[h(\boldsymbol{\theta},\boldsymbol{f}^{*}(\boldsymbol{m}))\right]<\inf_{\boldsymbol{f}\in{\cal F}:\lVert\boldsymbol{f}-\boldsymbol{f}^{*}\rVert\geq\epsilon}\mathbb{E}\left[h(\boldsymbol{\theta},\boldsymbol{f}(\boldsymbol{m}))\right]$
for any $\epsilon>0$.

2. For any $L<L'$, we have ${\cal F}_{L}\subseteq{\cal F}_{L'}\subseteq{\cal F}$,
and there exists a sequence of functions $\boldsymbol{f}_{L}\in{\cal F}_{L}$
such that $\lVert\boldsymbol{f}_{L}-\boldsymbol{f}^{*}\rVert\rightarrow0$.

3. $\mathbb{E}\left[h(\boldsymbol{\theta},\boldsymbol{f}(\boldsymbol{m}))\right]$
is continuous in $\boldsymbol{f}$ under the norm $\lVert\cdot\rVert$
on ${\cal F}$.

4. Each ${\cal F}_{L}$ is compact under $\lVert\cdot\rVert$.

5. $\sup_{\boldsymbol{f}\in{\cal F}_{L}}|C(\boldsymbol{f})-\mathbb{E}\left[h(\boldsymbol{\theta},\boldsymbol{f}(\boldsymbol{m}))\right]|\rightarrow0$
in probability as $L\rightarrow\infty$.\medskip{}

To see these conditions indeed give $\lVert\widehat{\boldsymbol{f}}_{L}-\boldsymbol{f}^{*}\rVert\rightarrow0$
in probability, we need to make a series of arguments. To ease notations,
let $Q(\boldsymbol{f})$ denote the population loss function $\mathbb{E}\left[h(\boldsymbol{\theta},\boldsymbol{f}(\boldsymbol{m}))\right]$.
Given any $\epsilon>0$, take $\delta\equiv\inf_{\boldsymbol{f}\in{\cal F}:\lVert\boldsymbol{f}-\boldsymbol{f}^{*}\rVert\geq\epsilon}Q(\boldsymbol{f})-Q(\boldsymbol{f}^{*})$.
By condition 1, we have $\delta>0$. By condition 2, we can find a
sequence $\boldsymbol{f}_{L}\in{\cal F}_{L}$ that converges to $\boldsymbol{f}^{*}$
in $\lVert\cdot\rVert$. Given the definition $\widehat{\boldsymbol{f}}_{L}\in\text{argmin}_{\boldsymbol{f}\in{\cal F}_{L}}C(\boldsymbol{f})$,
we must have $C(\widehat{\boldsymbol{f}}_{L})\leq C(\boldsymbol{f}_{L})$.
By condition 5, we have $\Pr[|C(\widehat{\boldsymbol{f}}_{L})-Q(\widehat{\boldsymbol{f}}_{L})|<\delta/3]\rightarrow1$
and $\Pr\left[|C(\boldsymbol{f}_{L})-Q(\boldsymbol{f}_{L})|<\delta/3\right]\rightarrow1$
as $L\rightarrow\infty$. Therefore, $\Pr[Q(\widehat{\boldsymbol{f}}_{L})<Q(\boldsymbol{f}_{L})+2\delta/3]\rightarrow1$.
By condition 1, 3, and the construction of $\boldsymbol{f}_{L}$,
for any sufficiently large $L$ we have $Q(\boldsymbol{f}_{L})\leq Q(\boldsymbol{f}^{*})+\delta/3$.
Thus, $\Pr[Q(\widehat{\boldsymbol{f}}_{L})<Q(\boldsymbol{f}^{*})+\delta]\rightarrow1$
for any $\delta>0$. Using the definition of $\delta$, we get $\Pr[Q(\widehat{\boldsymbol{f}}_{L})<\inf_{\boldsymbol{f}\in{\cal F}:\lVert\boldsymbol{f}-\boldsymbol{f}^{*}\rVert\geq\epsilon}Q(\boldsymbol{f})]\rightarrow1$.
With $\widehat{\boldsymbol{f}}_{L}\in{\cal F}_{L}\subseteq{\cal F}$,
the event of $Q(\widehat{\boldsymbol{f}}_{L})<\inf_{\boldsymbol{f}\in{\cal F}:\lVert\boldsymbol{f}-\boldsymbol{f}^{*}\rVert\geq\epsilon}Q(\boldsymbol{f})$
implies $\lVert\widehat{\boldsymbol{f}}_{L}-\boldsymbol{f}^{*}\rVert<\epsilon$.
Thus, we have $\Pr[\lVert\widehat{\boldsymbol{f}}_{L}-\boldsymbol{f}^{*}\rVert<\epsilon]\rightarrow1$.

\medskip{}

We now check these five conditions are satisfied. Condition 1 is directly
provided by assumption (i). It basically requires the population loss
$\mathbb{E}\left[h(\boldsymbol{\theta},\boldsymbol{f}(\boldsymbol{m}))\right]$
to have a unique minimum at $\boldsymbol{f}^{*}$, and it is not approachable
elsewhere.

Condition 2 is shown by \citet{hornik1989approximation} and \citet{white1990connectionist}.
In general, neural nets are dense sieves for the space of continuous
functions.

Condition 3 needs the population loss function to be continuous in
$\boldsymbol{f}$. Our assumption (ii) says $h$ is Lipschitz continuous
(because it is continuously differentiable over a compact set). Therefore,
there exist some $c>0$ such that $\left|h(\boldsymbol{\theta},\boldsymbol{f}(\boldsymbol{m}))-h(\boldsymbol{\theta},\boldsymbol{f}'(\boldsymbol{m}))\right|\leq c\lVert\boldsymbol{f}(\boldsymbol{m})-\boldsymbol{f}'(\boldsymbol{m})\rVert$
for any $(\boldsymbol{\theta},\boldsymbol{m})$. Thus, $\left|\mathbb{E}h(\boldsymbol{\theta},\boldsymbol{f}(\boldsymbol{m}))-\mathbb{E}h(\boldsymbol{\theta},\boldsymbol{f}'(\boldsymbol{m}))\right|\leq c\mathbb{E}\lVert\boldsymbol{f}(\boldsymbol{m})-\boldsymbol{f}'(\boldsymbol{m})\rVert\leq c\sqrt{\mathbb{E}\lVert\boldsymbol{f}(\boldsymbol{m})-\boldsymbol{f}'(\boldsymbol{m})\rVert^{2}}=c\lVert\boldsymbol{f}-\boldsymbol{f}'\rVert$
(the second step uses Jensen's inequality). Therefore, the population
loss is continuous in $\boldsymbol{f}$.

For condition 4, we note that the derivatives of the functions in
${\cal F}_{L}$ are bounded, and thus are Lipschitz continuous with
a common Lipschitz constant. By the Arzela-Ascoli theorem, any function
sequence in ${\cal F}_{L}$ must have a convergent subsequence. Thus,
${\cal F}_{L}$ is compact.

Condition 5 is a high-level condition. It can be implied by the lower-level
condition 3.5M in \citet{chen2007sieve}. We note that condition 5
does not inherit the specific metric between functions used in condition
3.5M. We will use the sup norm $\lVert\cdot\rVert_{\infty}$ as this
metric when applying condition 3.5M.

Condition 3.5M(i) requires i.i.d. samples for the computation of $C$.
This requirement is satisfied by how we construct the training examples
in NNE. In addition, the condition requires $\mathbb{E}\left[\mathrm{sup}_{\boldsymbol{f}\in{\cal F}_{L}}|h(\boldsymbol{\theta},\boldsymbol{f}(\boldsymbol{m}))|\right]<\infty$
for all $L$. This requirement is provided for by our assumption (ii),
which says that $h$ is continuous and thus bounded over the compact
$\Theta\times\Delta$.

Condition 3.5M(ii) is satisfied if there is a constant $c>0$ such
that for any $\boldsymbol{\theta}$ and $\boldsymbol{m}$, we have
$|h(\boldsymbol{\theta},\boldsymbol{f}(\boldsymbol{m}))-h(\boldsymbol{\theta},\boldsymbol{f}'(\boldsymbol{m}))|\leq c\lVert\boldsymbol{f}-\boldsymbol{f}'\rVert_{\infty}$.
To get this result, note the Euclidean distance between $\boldsymbol{f}(\boldsymbol{m})$
and $\boldsymbol{f}'(\boldsymbol{m})$ is bounded by $\sqrt{\dim(\Delta)}\times\lVert\boldsymbol{f}-\boldsymbol{f}'\rVert_{\infty}$.
In addition, assumption (ii) says $h$ is Lipschitz continuous. Thus,
we only need to take $c$ as the product of $\sqrt{\dim(\Delta)}$
and the Lipschitz constant of $h$.

Condition 3.5M(iii) says the number of balls required to cover ${\cal F}_{L}$
cannot grow too fast. This is satisfied with any sufficiently slow
rates of $u_{L}$ and $b_{L}$.

We have checked all the conditions for $\lVert\widehat{\boldsymbol{f}}_{L}-\boldsymbol{f}^{*}\rVert$
to converge to zero in probability.
\end{proof}
\medskip{}

\begin{proof}
of \textbf{Proposition} \ref{prop: mean}: We apply Lemma \ref{lem:sieve},
with the target function as $\boldsymbol{f}^{*}=\mathbb{E}(\boldsymbol{\theta}|\boldsymbol{m})$.
We can take $\Delta$ as any compact set that contains $\mathrm{conv}(\Theta)$
in its interior, which is always possible given that $\Theta$ is
compact. The reason for taking the convex hull is to accommodate cases
where some dimensions of $\Theta$ are discrete so that $\mathbb{E}(\boldsymbol{\theta}|\boldsymbol{m})$
may take values outside $\Theta$. One simplest choice of $\Delta$
is $[\underline{\theta},\overline{\theta}]^{p}$, where $p$ denotes
the dimension of $\boldsymbol{\theta}$ and $\underline{\theta}<\theta_{k}<\overline{\theta}$
for all $\boldsymbol{\theta}\in\Theta$ and all $k=1,...,p$.

Assumption (ii) of Lemma \ref{lem:sieve} is immediately provided
by the square-error form of $h$, that is, $h(\boldsymbol{\theta},\boldsymbol{f})=\sum_{k}(\theta_{k}-f_{k})^{2}$.

As to assumption (i), we need to show $\boldsymbol{f}^{*}$ minimizes
the population loss in ${\cal F}$ and further, $\lVert\boldsymbol{f}-\boldsymbol{f}^{*}\rVert\geq\epsilon$
can bound the population loss's value at $\boldsymbol{f}$ away from
this minimum. It suffices to consider the case where $\boldsymbol{\theta}$
is single-dimensional, because the population loss simply sums across
each dimension. Note
\[
\mathbb{E}(\theta-f(\boldsymbol{m}))^{2}=\int\mathbb{E}\left[(\theta-f(\boldsymbol{m}))^{2}|\boldsymbol{m}\right]dP(\boldsymbol{m}).
\]
Take any $f\neq f^{*}$ . We have
\begin{align*}
\mathbb{E}\left[(\theta-f(\boldsymbol{m}))^{2}|\boldsymbol{m}\right] & -\mathbb{E}\left[(\theta-f^{*}(\boldsymbol{m}))^{2}|\boldsymbol{m}\right]\\
= & \mathbb{E}\left[-2\theta f(\boldsymbol{m})+2\theta f^{*}(\boldsymbol{m})+f(\boldsymbol{m})^{2}-f^{*}(\boldsymbol{m})^{2}|\boldsymbol{m}\right]\\
= & -2f^{*}(\boldsymbol{m})f(\boldsymbol{m})+2f^{*}(\boldsymbol{m})^{2}+f(\boldsymbol{m})^{2}-f^{*}(\boldsymbol{m})^{2}\\
= & \left[f(\boldsymbol{m})-f^{*}(\boldsymbol{m})\right]^{2}.
\end{align*}
The second equality uses the fact $f^{*}=\mathbb{E}(\theta|\boldsymbol{m})$.
As a result,
\[
\mathbb{E}(\theta-f(\boldsymbol{m}))^{2}-\mathbb{E}(\theta-f^{*}(\boldsymbol{m}))^{2}=\lVert f-f^{*}\rVert^{2}.
\]
Thus, the difference between the population loss at $f$ and the the
population loss at $f^{*}$ is exactly $\lVert f-f^{*}\rVert^{2}$.
In particular, $\inf_{f\in{\cal F}:\lVert f-f^{*}\rVert\geq\epsilon}\mathbb{E}\left[h(\theta,f(\boldsymbol{m}))\right]=\mathbb{E}\left[h(\theta,f^{*}(\boldsymbol{m}))\right]+\epsilon^{2}$.
Thus, assumption (i) is satisfied.
\end{proof}
\medskip{}

\begin{proof}
of \textbf{Proposition} \ref{prop: cov} - part (i): We apply Lemma
\ref{lem:sieve}, with $\boldsymbol{f}^{*}=\left[\mathbb{E}(\boldsymbol{\theta}|\boldsymbol{m}),\mathbb{V}\mathrm{ar}(\boldsymbol{\theta}|\boldsymbol{m})\right]$.
Again let $p$ be the dimension of $\boldsymbol{\theta}$. Because
$\mathbb{V}\mathrm{ar}(\boldsymbol{\theta}|\boldsymbol{m})$ is assumed
to be bounded away from zero and $\Theta$ is bounded, we may choose
$\underline{v},\overline{v}>0$ such that $\underline{v}<\mathbb{V}\mathrm{ar}(\theta_{k}|\boldsymbol{m})<\overline{v}$
for all $k$. Let $\underline{\theta}$ and $\overline{\theta}$ be
defined as in the proof of Proposition \ref{prop: mean}, we may choose
$\Delta$ to be $[\underline{\theta},\overline{\theta}]^{p}\times[\underline{v},\overline{v}]^{p}$.

The individual loss function is $h(\boldsymbol{\theta},\boldsymbol{f})=\sum_{k}\log(V_{k})+V_{k}^{-1}(\theta_{k}-\mu_{k})^{2}$,
where $\{\mu_{k},V_{k}\}_{k=1}^{p}$ are collected in $\boldsymbol{f}$.
Assumption (ii) of Lemma \ref{lem:sieve} is satisfied with this choice
of $h$.

As to assumption (i), we use a similar argument as in the proof of
Proposition \ref{prop: mean}. As before, we consider the case where
$p=1$, so that we may write $\boldsymbol{f}\equiv(\mu,V)$. The cases
with $p>1$ follow because the population loss function again simply
sums across each dimension $k$. Fix a small $\epsilon>0$. Our argument
starts by noting
\[
\mathbb{E}h(\theta,\boldsymbol{f}(\boldsymbol{m}))=\int\mathbb{E}\left[h(\theta,\boldsymbol{f}(\boldsymbol{m}))|\boldsymbol{m}\right]dP(\boldsymbol{m}).
\]
We first examine the part inside the integral. We have
\begin{align}
 & \mathbb{E}\left[h(\theta,\boldsymbol{f}(\boldsymbol{m}))|\boldsymbol{m}\right]-\mathbb{E}\left[h(\theta,\boldsymbol{f}^{*}(\boldsymbol{m}))|\boldsymbol{m}\right]\nonumber \\
 & \qquad=\mathbb{E}\left[\log(V(\boldsymbol{m}))+\frac{(\theta-\mu(\boldsymbol{m}))^{2}}{V(\boldsymbol{m})}-\log(V^{*}(\boldsymbol{m}))-\frac{(\theta-\mu^{*}(\boldsymbol{m}))^{2}}{V^{*}(\boldsymbol{m})}\:\Big|\boldsymbol{m}\right]\nonumber \\
 & \qquad=\frac{V^{*}(\boldsymbol{m})}{V(\boldsymbol{m})}-\log\left[\frac{V^{*}(\boldsymbol{m})}{V(\boldsymbol{m})}\right]-1+\frac{1}{V(\boldsymbol{m})}\cdot(\mu^{*}(\boldsymbol{m})-\mu(\boldsymbol{m}))^{2}.\label{eq:proof.1}
\end{align}
The last equality uses the definitions $\mu^{*}(\boldsymbol{m})=\mathbb{E}(\theta|\boldsymbol{m})$
and $V^{*}(\boldsymbol{m})=\mathbb{V}\mathrm{ar}(\theta|\boldsymbol{m})$.
We want to show that the integral of (\ref{eq:proof.1}) is bounded
away from zero if $\lVert\boldsymbol{f}-\boldsymbol{f}^{*}\rVert\geq\epsilon$.
We do this in two steps.

First, let $Q:\Delta^{2}\rightarrow\mathbb{R}$ be defined as $Q(\boldsymbol{t},\boldsymbol{t}')=t_{2}'/t_{2}-\log(t_{2}'/t_{2})-1+(1/t_{2})\cdot(t_{1}'-t_{1})^{2}$.
Then (\ref{eq:proof.1}) equals $Q(\boldsymbol{f}(\boldsymbol{m}),\boldsymbol{f}^{*}(\boldsymbol{m}))$.
Define $d(\boldsymbol{m})\equiv\inf_{\boldsymbol{t}\in\Delta:\lVert\boldsymbol{t}-\boldsymbol{f}^{*}(\boldsymbol{m})\rVert\geq\epsilon/2}Q(\boldsymbol{t},\boldsymbol{f}^{*}(\boldsymbol{m}))$.
Note $d$ is positive because $Q$ reaches its minimum zero only when
$\boldsymbol{t}=\boldsymbol{t}'$. In addition, by Berge's theorem,
$d$ is a continuous function. This, together with the compactness
of ${\cal M}$, gives $\delta\equiv\inf_{\boldsymbol{m}\in{\cal M}}d(\boldsymbol{m})>0$.
We have the result that for any $\boldsymbol{f}$ and $\boldsymbol{m}$,
$\lVert\boldsymbol{f}(\boldsymbol{m})-\boldsymbol{f}^{*}(\boldsymbol{m})\rVert\geq\epsilon/2$
implies (\ref{eq:proof.1}) is no less than $\delta$.

Second, let $A\equiv\{\boldsymbol{m}\in{\cal M}:\lVert\boldsymbol{f}(\boldsymbol{m})-\boldsymbol{f}^{*}(\boldsymbol{m})\rVert\geq\epsilon/2\}$.
We have $\lVert\boldsymbol{f}-\boldsymbol{f}^{*}\rVert^{2}=\int\lVert\boldsymbol{f}(\boldsymbol{m})-\boldsymbol{f}^{*}(\boldsymbol{m})\rVert^{2}dP(\boldsymbol{m})\leq P(A)c^{2}+(1-P(A))\epsilon^{2}/4$,
where $c$ denotes an upper bound for $\lVert\boldsymbol{f}(\boldsymbol{m})-\boldsymbol{f}^{*}(\boldsymbol{m})\rVert$
implied by the compactness of $\Delta$. A result of this inequality
is that $P(A)$ cannot be too small if $\lVert\boldsymbol{f}-\boldsymbol{f}^{*}\rVert$
is not too small. More precisely, there exists some $\tau>0$ such
that for any $\boldsymbol{f}$, $\lVert\boldsymbol{f}-\boldsymbol{f}^{*}\rVert\geq\epsilon$
implies $P(A)\geq\tau$.

Now, pick any $\boldsymbol{f}$ with $\lVert\boldsymbol{f}-\boldsymbol{f}^{*}\rVert\geq\epsilon$,
we have
\begin{align*}
\mathbb{E}\left[h(\theta,\boldsymbol{f}(\boldsymbol{m}))\right] & -\mathbb{E}\left[h(\theta,\boldsymbol{f}^{*}(\boldsymbol{m}))\right]\\
 & =\int\left\{ \mathbb{E}\left[h(\theta,\boldsymbol{f}(\boldsymbol{m}))|\boldsymbol{m}\right]-\mathbb{E}\left[h(\theta,\boldsymbol{f}^{*}(\boldsymbol{m}))|\boldsymbol{m}\right]\right\} dP(\boldsymbol{m})\\
 & \geq\int_{A}\delta dP(\boldsymbol{m})\\
 & \geq\tau\delta.
\end{align*}
In words, the population loss function at $\boldsymbol{f}$ is at
least $\tau\delta>0$ larger than the population loss at $\boldsymbol{f}^{*}$,
where neither $\tau$ or $\delta$ depends on $\boldsymbol{f}$. Therefore,
assumption (i) is satisfied.

\bigskip{}

Part (ii): We again apply Lemma \ref{lem:sieve}. By the condition
of the proposition, we can find some $\underline{\lambda}>0$ such
that the smallest eigenvalue of $\mathbb{C}\mathrm{ov}(\boldsymbol{\theta}|\boldsymbol{m})$
is larger than $\underline{\lambda}$. Let $\nabla$ denote the set
of all possible values of covariance matrix $\boldsymbol{V}$ such
that $\underline{v}\leq\text{diag}(\boldsymbol{V})\leq\overline{v}$
and the smallest eigenvalue of $\boldsymbol{V}$ is no less than $\underline{\lambda}$.
Note $\nabla$ specified as such is a compact subset of the convex
cone for positive definite matrices. We can then take $\Delta=[\underline{\theta},\overline{\theta}]^{p}\times\nabla$.

The specification of $h(\boldsymbol{\theta},\boldsymbol{f})=\log(|\boldsymbol{V}|)+(\boldsymbol{\theta}-\boldsymbol{\mu})'\boldsymbol{V}^{-1}(\boldsymbol{\theta}-\boldsymbol{\mu})$,
where $\boldsymbol{f}=(\boldsymbol{\mu},\boldsymbol{V})$, satisfies
assumption (ii) in Lemma \ref{lem:sieve}. Note, in particular, that
matrix inverse is a continuously differentiable operation on positive
definite matrices.

The proof to show assumption (i) is satisfied can use a similar argument
as the proof for part (i). The target function $\boldsymbol{f}^{*}$
is $\boldsymbol{\mu}^{*}=\mathbb{E}(\boldsymbol{\theta}|\boldsymbol{m})$
and $\boldsymbol{V}^{*}=\mathbb{C}\mathrm{ov}(\boldsymbol{\theta}|\boldsymbol{m})$.
Again we note the population loss can be written as
\[
\mathbb{E}h(\boldsymbol{\theta},\boldsymbol{f}(\boldsymbol{m}))=\int\mathbb{E}\left[h(\boldsymbol{\theta},\boldsymbol{f}(\boldsymbol{m}))|\boldsymbol{m}\right]dP(\boldsymbol{m}),
\]
where the integrand satisfies, with $\boldsymbol{f}=(\boldsymbol{\mu},\boldsymbol{V})$,
\begin{align}
\mathbb{E}\left[h(\boldsymbol{\theta},\boldsymbol{f}(\boldsymbol{m}))\big|\boldsymbol{m}\right] & -\mathbb{E}\left[h(\boldsymbol{\theta},\boldsymbol{f}^{*}(\boldsymbol{m}))\big|\boldsymbol{m}\right]\nonumber \\
=\; & \mathbb{E}\left[\log(|\boldsymbol{V}(\boldsymbol{m})|)+(\boldsymbol{\theta}-\boldsymbol{\mu}(\boldsymbol{m}))'\boldsymbol{V}(\boldsymbol{m})^{-1}(\boldsymbol{\theta}-\boldsymbol{\mu}(\boldsymbol{m}))\:\big|\:\boldsymbol{m}\right]\nonumber \\
 & \quad-\mathbb{E}\left[\log(|\boldsymbol{V}^{*}(\boldsymbol{m})|)+(\boldsymbol{\theta}-\boldsymbol{\mu}^{*}(\boldsymbol{m}))'\boldsymbol{V}^{*}(\boldsymbol{m})^{-1}(\boldsymbol{\theta}-\boldsymbol{\mu}^{*}(\boldsymbol{m}))\:\big|\:\boldsymbol{m}\right].\label{eq:proof.2}
\end{align}
Because $\boldsymbol{\theta}$ enters the above expectations in a
quadratic way, (\ref{eq:proof.2}) depends on the first and second
moments, but not the higher moments or distributional form, of $P(\boldsymbol{\theta}|\boldsymbol{m})$.
These first and second moments are given by $\boldsymbol{f}^{*}(\boldsymbol{m})$.
Therefore, we can write a function $Q:\Delta^{2}\rightarrow\mathbb{R}$
like in the proof of part (i) such that $Q(\boldsymbol{f}(\boldsymbol{m}),\boldsymbol{f}^{*}(\boldsymbol{m}))$
equals (\ref{eq:proof.2}) for any $\boldsymbol{f}$ and $\boldsymbol{m}$.
This $Q$ does not depend on the distributional form of $P(\boldsymbol{\theta}|\boldsymbol{m})$,
and it is continuous over $\Delta^{2}$ (note both the matrix determinant
and matrix inverse are continuous operations).

In addition, $\min_{\boldsymbol{t}\in\Delta}Q(\boldsymbol{t},\boldsymbol{f}^{*}(\boldsymbol{m}))=0$
and is achieved at $\boldsymbol{t}=\boldsymbol{f}^{*}(\boldsymbol{m})$.
To see this, suppose $P(\boldsymbol{\theta}|\boldsymbol{m})$ is normal
for a moment. By (\ref{eq:proof.2}), we see $Q(\boldsymbol{t},\boldsymbol{f}^{*}(\boldsymbol{m}))$
is the Kullback--Leibler divergence from a normal density parameterized
by $\boldsymbol{t}\in\Delta$ to the normal $P(\boldsymbol{\theta}|\boldsymbol{m})$.
As a result, $Q(\boldsymbol{t},\boldsymbol{f}^{*}(\boldsymbol{m}))$
as a function of $\boldsymbol{t}$ is always positive except when
$\boldsymbol{t}$ takes the mean and covariance of $P(\boldsymbol{\theta}|\boldsymbol{m})$,
that is, $\boldsymbol{t}=\boldsymbol{f}^{*}(\boldsymbol{m})$. However,
because $Q$ does not depend on the distributional form of $P(\boldsymbol{\theta}|\boldsymbol{m})$,
this result holds for non-normal $P(\boldsymbol{\theta}|\boldsymbol{m})$
as well.

With these properties of $Q$ established, the rest of the argument
follows that in the proof for part (i). So we shall not repeat it
here.
\end{proof}
\medskip{}

\begin{proof}
of \textbf{Proposition} \ref{prop: general}: We apply Lemma \ref{lem:sieve},
with $\Delta=\Gamma$. With this choice of $\Delta$, the target $\boldsymbol{f}^{*}$
as defined in the proposition is a member of ${\cal F}$ iff it is
continuous. We will show $\boldsymbol{f}^{*}$ is continuous below.

Assumption (ii) of Lemma \ref{lem:sieve} is satisfied with $h(\boldsymbol{\theta},\boldsymbol{f}(\boldsymbol{m}))=-\log\phi(\boldsymbol{\theta};\boldsymbol{f}(\boldsymbol{m}))$
and condition (iii) of the proposition.

We proceed to verify assumption (i). To ease notation, let $Q(\boldsymbol{\gamma},\boldsymbol{m})\equiv\mathbb{E}\left[-\log\phi(\boldsymbol{\theta};\boldsymbol{\gamma})\big|\boldsymbol{m}\right]$.
By the definition of Kullback--Leibler divergence, we have
\[
{\cal KL}\left[P(\boldsymbol{\theta}|\boldsymbol{m})\:\lVert\:\phi(\boldsymbol{\theta};\boldsymbol{\gamma})\right]=Q(\boldsymbol{\gamma},\boldsymbol{m})-\mathbb{E}\left[-\log P(\boldsymbol{\theta}|\boldsymbol{m})\big|\boldsymbol{m}\right].
\]
Note the second term on the right side does not involve $\boldsymbol{\gamma}$.
Therefore, by condition (v), we know $\boldsymbol{f}^{*}(\boldsymbol{m})$
is equal to $\mathrm{argmin}_{\gamma\in\Gamma}Q(\boldsymbol{\gamma},\boldsymbol{m})$
for every $\boldsymbol{m}$. Now, with the continuity provided by
condition (iv), Berge's theorem says that $\boldsymbol{f}^{*}$ is
a continuous function.

Next, fix any small $\epsilon>0$. Consider the following function
$d$:
\[
d(\boldsymbol{m})\equiv\inf_{\boldsymbol{\gamma}\in\Gamma:\lVert\boldsymbol{\gamma}-\boldsymbol{f}^{*}(\boldsymbol{m})\rVert\geq\epsilon/2}Q(\boldsymbol{\gamma},\boldsymbol{m})-Q(\boldsymbol{f}^{*}(\boldsymbol{m}),\boldsymbol{m}).
\]
Again by Berge's theorem, we know $d$ is continuous. Because $\boldsymbol{f}^{*}(\boldsymbol{m})$
is the unique point in $\Gamma$ that minimizes $Q(\boldsymbol{\gamma},\boldsymbol{m})$,
we also have $d>0$. As a result, $\inf_{\boldsymbol{m}\in{\cal M}}d(\boldsymbol{m})$
is attainable and positive, which we denote as $\delta>0$.

Let $A\equiv\{\boldsymbol{m}\in{\cal M}:\lVert\boldsymbol{f}(\boldsymbol{m})-\boldsymbol{f}^{*}(\boldsymbol{m})\rVert\geq\epsilon/2\}$.
Using the same argument in the proof of Proposition \ref{prop: cov}
part (i), we can find some $\tau>0$ such that for any $\boldsymbol{f}$,
$\lVert\boldsymbol{f}-\boldsymbol{f}^{*}\rVert\geq\epsilon$ implies
$P(A)\geq\tau$.

Now for any $\boldsymbol{f}$ such that $\lVert\boldsymbol{f}-\boldsymbol{f}^{*}\rVert\geq\epsilon$,
we have
\begin{align*}
\mathbb{E}\left[h(\boldsymbol{\theta},\boldsymbol{f}(\boldsymbol{m}))\right] & -\mathbb{E}\left[h(\boldsymbol{\theta},\boldsymbol{f}^{*}(\boldsymbol{m}))\right]\\
 & =\int\Big\{ Q\left[\boldsymbol{f}(\boldsymbol{m}),\boldsymbol{m}\right]-Q\left[\boldsymbol{f}^{*}(\boldsymbol{m}),\boldsymbol{m}\right]\Big\} dP(\boldsymbol{m})\\
 & \geq\int_{A}\Big\{ Q\left[\boldsymbol{f}(\boldsymbol{m}),\boldsymbol{m}\right]-Q\left[\boldsymbol{f}^{*}(\boldsymbol{m}),\boldsymbol{m}\right]\Big\} dP(\boldsymbol{m})\\
 & \geq\int_{A}d(\boldsymbol{m})\cdot dP(\boldsymbol{m})\\
 & \geq\tau\delta.
\end{align*}
In words, with $\lVert\boldsymbol{f}-\boldsymbol{f}^{*}\rVert\geq\epsilon$,
the population loss function at $\boldsymbol{f}$ is at least $\tau\delta>0$
larger than the population loss at $\boldsymbol{f}^{*}$. Thus, assumption
(i) is satisfied.
\end{proof}
\medskip{}

\newpage
\
\newpage

\appendix
\newpage
{\small}

\section*{Funding and Competing Interests}

All authors certify that they have no affiliations with or involvement
in any organization or entity with any financial interest or non-financial
interest in the subject matter or materials discussed in this manuscript.
The authors have no funding to report.

\appendix
\newpage
\renewcommand*{\bibfont}{\small}

\bibliographystyle{econ}
\bibliography{reference}

\end{document}